%% file: main.tex
\documentclass[letterpaper,journal]{IEEEtran}
\input{macros.tex}
\usepackage{multibib}
\newcites{SM}{SM References}

\begin{document}
\title{Sheaf-theoretic Signal Processing on Graphs:\\Spectral Theory,  Filtering, and Sampling}
\author{Gabriele D'Acunto~\IEEEmembership{(Member,~IEEE~SPS)}, Leonardo Di Nino~\IEEEmembership{(Student Member,~IEEE)}, \\Paolo Di Lorenzo~\IEEEmembership{(Senior Member,~IEEE)}, Sergio Barbarossa~\IEEEmembership{(Life Fellow,~IEEE)}
\thanks{The Authors are with the Department of Information Engineering, Electronics, and Telecommunications, Sapienza University of Rome, 00184 Rome, Italy (e-mails: \{gabriele.dacunto, leonardo.dinino, paolo.dilorenzo, sergio.barbarossa\}@uniroma1.it).
All authors are also with the National Inter-University Consortium for Telecommunications (CNIT), Parma, Italy. The work was supported by the SNS JU project 6G-GOALS~\cite{strinati2024goal} under the EU’s Horizon program Grant Agreement No 101139232, and by Huawei Technology France SASU under Grant N. Tg20250616041.
} \vspace{-.6cm}
}


\maketitle
\input{sections/abstract}


\input{sections/0-introduction}
\input{sections/1-primer-on-network-sheaf.tex}
\input{sections/2-structured-sheaf-signals.tex}
\input{sections/3-sheaf-spectral-theory}
\input{sections/4-sheaf_filter.tex}
\input{sections/5-sheaf-sampling}
\input{sections/6-applications.tex}
\input{sections/7-conclusions}

\balance
\bibliographystyle{IEEEtran}
\bibliography{bibliography}

\clearpage

\section*{Supplementary Material for\\ \enquote{Sheaf-theoretic Signal Processing on Graphs: Spectral Theory,  Filtering, and Sampling}}

\opensupplement

\input{apps/linear-algebra-with-metric-tensors.tex}
\input{apps/proofs.tex}
\input{apps/spectral-multiplicity}

\balance
\bibliographystyleSM{IEEEtran}
\bibliographySM{bibliographySM}

\closesupplement

\end{document}

%% file: macros.tex
\usepackage[english]{babel} 
\usepackage[expansion=false]{microtype}      
\usepackage{csquotes}       
\usepackage[normalem]{ulem} 
\usepackage{xspace}         
\usepackage{textcomp}       
\usepackage{enumitem}
\usepackage{cite}

\usepackage{type1cm}        
\usepackage{bold-extra}     
\usepackage{bm}             
\usepackage{bbm}            
\usepackage{bbold}          
\usepackage{siunitx}        
\usepackage{xfrac}          

\usepackage{amsmath,amsfonts,amssymb,mathtools} 
\usepackage{amsthm}
\usepackage{physics}        
\usepackage{thmtools,thm-restate} 

\usepackage{booktabs}       
\usepackage{multirow}       
\usepackage{array}          
\usepackage{rotating}

\usepackage{graphicx}       
\usepackage[style=base, tableposition=top]{caption} 
\usepackage[caption=false,font=normalsize,labelfont=sf,textfont=sf]{subfig} 
\usepackage{tikz}           
\usepackage{tikz-cd} \usetikzlibrary{arrows.meta}
\usetikzlibrary{patterns, positioning, arrows, bayesnet, calc} 

\usepackage{tikz-cd}        
\usepackage{pgfplots}       
\usepgfplotslibrary{patchplots}
\pgfplotsset{compat=1.15}   
\usepackage[most]{tcolorbox}      
\usepackage{wrapfig}
\usepackage{algorithm}
\usepackage{algorithmic}

\usepackage[bookmarks,colorlinks=true]{hyperref} 
\usepackage[capitalize]{cleveref} 

\usepackage{balance}        
\usepackage{hyphenat}       
\usepackage[show]{chato-notes} 
\usepackage{stfloats}       
\usepackage{verbatim}       
\usepackage{cuted}

\newcommand{\spara}[1]{\smallskip\noindent\textbf{#1}}

{\par\egroup\vskip 0.25ex}

\newenvironment{squishenum}
{\begin{enumerate}[label=\emph{\roman*}.,
  itemsep=0pt,
  parsep=3pt,
  topsep=3pt,
  partopsep=0pt,
  leftmargin=1.5em]}
{\end{enumerate}}

\usepackage[dvipsnames]{xcolor}
\hypersetup{
   colorlinks=true,
   linkcolor={red!50!black},
   filecolor={green!50!black},
   citecolor={green!50!black}, 
   urlcolor={blue!80!black},
}

\definecolor{mypurple}{RGB}{254, 68, 218}
\definecolor{myred}{HTML}{E13D66}
\definecolor{mycyan}{HTML}{70D7D0}
\definecolor{mylightblue}{HTML}{2274A5}
\definecolor{mydarkblue}{HTML}{0C0A3E}
\definecolor{mydarkviolet}{HTML}{7B2D8B}
\definecolor{myorange}{HTML}{F4A261}

\newcommand{\mylightblue}[1]{\textcolor{mylightblue}{#1}}

\newcommand{\mydarkviolet}[1]{\textcolor{mydarkviolet}{#1}}
\newcommand{\myorange}[1]{\textcolor{myorange}{#1}}

\theoremstyle{plain}
\newtheorem{theorem}{Theorem}[section]
\newtheorem{proposition}[theorem]{Proposition}

\newtheorem{corollary}[theorem]{Corollary}

\theoremstyle{remark}
\newtheorem{remark}{Remark}
\newtheorem{definition}{Definition}

\newtheorem{example}{Example}

\crefname{theorem}{Thm.}{Thms.}
\crefname{proposition}{Prop.}{Props.}
\crefname{lemma}{lem.}{lems.}
\crefname{corollary}{Cor.}{Cors.}
\crefname{definition}{Def.}{Defs.}
\crefname{section}{Sec.}{Secs.}
\crefname{figure}{Fig.}{Figs.}
\crefname{problem}{Prob.}{Probs.}
\crefname{appendix}{App.}{Apps.}
\crefname{equation}{Eq.}{Eqs.}
\crefname{algorithm}{Alg.}{Algs.}
\crefname{table}{Tab.}{Tabs.}
\crefname{example}{Ex.}{Exs.}
\crefname{remark}{Rem.}{Rems.}

\newcommand{\opensupplement}{
    
    \newcounter{SIsection}
    \renewcommand{\theSIsection}{Supplementary Section \arabic{section}}
    
    \crefalias{appendix}{supp_sec}
    \crefalias{figure}{supp_fig}
    \crefalias{subfigure}{supp_subfig}
    \crefalias{table}{supp_table}

    \setcounter{theorem}{0}
    \setcounter{equation}{0}
    \setcounter{figure}{0}
    \setcounter{table}{0}
    \setcounter{page}{1}
    \setcounter{section}{0}
    
    \renewcommand{\figurename}{Supplementary Figure}
    \renewcommand{\tablename}{Supplementary Table}

    \renewcommand{\thesection}{Supplementary Section \arabic{section}}
    \renewcommand{\theHsection}{Supplementary Section \arabic{section}}
    \renewcommand{\thepage}{\arabic{page}}
    
}

\newcommand{\closesupplement}{
    \renewcommand{\thepage}{\arabic{page}}
}

\DeclareMathOperator*{\argmin}{arg\,min}

\newcommand{\im}[1]{\ensuremath{\mathrm{Im}\left(#1\right)}\xspace}

\newcommand{\Hilb}{\ensuremath{\mathsf{Hilb}_{\reall}}\xspace}

\newcommand{\ns}{\ensuremath{F}\xspace}
\newcommand{\restrictionmap}[2]{\ensuremath{\mathbf{F}_{ #1 \trianglelefteq #2 }}\xspace}
\newcommand{\extensionmap}[2]{\ensuremath{\bar{\mathbf{F}}_{ #1 \trianglelefteq #2 }}\xspace}

\newcommand{\nscoeff}{\ensuremath{J}\xspace}
\newcommand{\restrictionmapscoeff}[2]{\ensuremath{\mathbf{J}_{#1 \trianglelefteq #2}}\xspace}
\newcommand{\nsrep}{\ensuremath{J}\xspace}

\newcommand{\pd}{\ensuremath{\mathcal{S}_{++}}\xspace}

\newcommand{\reall}{\ensuremath{\mathbb{R}}\xspace}

\newcommand{\starr}[1]{\ensuremath{\star\left({#1}\right)}\xspace}
\newcommand{\cochainspace}[3]{\ensuremath{\mathcal{C}^{#1}\left({#2}; {#3}\right)}\xspace}
\renewcommand{\cochainspace}[3]{\mathcal{C}^{#1}\left(#2;#3\right)}
\newcommand{\gsspace}[2]{\ensuremath{\Gamma\left({#1}; {#2}\right)}\xspace}

\newcommand{\edgeset}{\ensuremath{\mathcal{E}}\xspace}
\newcommand{\graph}{\ensuremath{\mathcal{G}}\xspace}
\newcommand{\subgraph}{\ensuremath{\mathcal{U}}\xspace}

\newcommand{\vertexset}{\ensuremath{\mathcal{N}}\xspace}

\newcommand{\Eprod}[3]{\ensuremath{\langle #1, \, #2 \rangle_{#3}}\xspace}
\renewcommand{\Eprod}[3]{\left\langle #1,\,#2\right\rangle_{#3}}

\newcommand{\stiefel}[2]{\ensuremath{\mathrm{St}({#1},{#2})}\xspace}

\newcommand{\specialO}[1]{\ensuremath{\mathrm{SO}(#1)}\xspace}

\newcommand{\zeros}{\ensuremath{\boldsymbol{0}}\xspace}

\newcommand{\vecu}{\ensuremath{\mathbf{u}}\xspace}
\newcommand{\vecv}{\ensuremath{\mathbf{v}}\xspace}
\newcommand{\vecw}{\ensuremath{\mathbf{w}}\xspace}
\newcommand{\vecs}{\ensuremath{\mathbf{s}}\xspace}
\newcommand{\vecy}{\ensuremath{\mathbf{y}}\xspace}

\newcommand{\x}{\ensuremath{\mathbf{x}}\xspace}
\newcommand{\y}{\ensuremath{\mathbf{y}}\xspace}
\newcommand{\w}{\ensuremath{\mathbf{w}}\xspace}

\newcommand{\m}{\ensuremath{\mathbf{m}}\xspace}

\newcommand{\A}{\ensuremath{\mathbf{A}}\xspace}
\newcommand{\B}{\ensuremath{\mathbf{B}}\xspace}

\newcommand{\D}{\ensuremath{\mathbf{D}}\xspace}
\newcommand{\G}{\ensuremath{\mathbf{G}}\xspace}
\newcommand{\Hmet}{\ensuremath{\mathbf{H}}\xspace}
\newcommand{\myH}{\ensuremath{\mathbf{H}}\xspace}
\newcommand{\identity}{\ensuremath{\mathbf{I}}\xspace}
\newcommand{\eye}[1]{\ensuremath{\identity_{#1}}\xspace}

\newcommand{\M}{\ensuremath{\mathbf{M}}\xspace}
\newcommand{\myO}{\ensuremath{\mathbf{O}}\xspace}

\newcommand{\T}{\ensuremath{\mathbf{T}}\xspace}
\newcommand{\U}{\ensuremath{\mathbf{U}}\xspace}
\newcommand{\V}{\ensuremath{\mathbf{V}}\xspace}

\newcommand{\bdO}{\ensuremath{\mathbb{O}}\xspace}

\renewcommand{\Lsh}{\mathbf{L}_{\ns}}
\newcommand{\Lshrep}{\mathbf{L}_{\nsrep}}

%% file: sections/abstract.tex

\begin{abstract}
Modern sensing, communication, and learning systems generate heterogeneous network signals, with local data differing in dimension, modality, and geometric structure. Processing such data requires a mathematical framework capable of simultaneously modeling heterogeneous local signal spaces and the transformations relating them. Network sheaves provide such a framework by associating local vector spaces with network entities and linear restriction maps with their interactions.

This is the first paper to develop a unified sheaf signal processing (SSP) framework on network sheaves, extending the fundamental operations of signal processing, namely spectral analysis, filtering, and sampling, to heterogeneous local spaces. Unlike graph and topological signal processing, where signals are modeled over a common vector space, SSP jointly models heterogeneous local signal spaces and the linear transformations relating neighboring spaces through restriction maps.

We define the Sheaf Fourier Transform (SFT), whose frequencies quantify signal inconsistency induced by the network topology, the restriction maps, and the local geometry. Building on this representation, we develop polynomial sheaf filters and formulate sampling as the joint selection of network nodes and intra-node components. We derive perfect recovery conditions for bandlimited sheaf signals and propose a greedy sampling-set design algorithm. To incorporate application-dependent signal models, including different bases, dictionaries, and learned embeddings, we introduce representation sheaves and characterize the natural transformations that preserve spectral properties and guarantee interoperability across representations. Experiments on synthetic, motion-capture, and financial datasets validate the proposed framework and demonstrate consistent improvements over canonical graph signal processing baselines.
\end{abstract}

%% file: sections/0-introduction.tex

\section{Introduction}\label{sec:introduction}

In recent years, Graph Signal Processing (GSP) has emerged as a powerful extension of classical SP to irregular domains, where signals are supported on the vertices of a graph rather than on regular Euclidean grids~\cite{shuman2013emerging,ortega2018graph}. 
By encoding the connectivity of the domain into graph shift operators, such as the adjacency matrix or the graph Laplacian, GSP provides a principled framework to define graph Fourier transforms (GFTs), localized filters, sampling strategies, and learning methods for networked data~\cite{sandryhaila2013discrete,shuman2013emerging,ortega2018graph,tsitsvero2016signals,mateos2019connecting,sardellitti2017graph,di2018adaptive}. 
The success of GSP stems from the fact that the processing operators are intrinsically adapted to the relational structure of the data. 
However, graph models are inherently restricted to pairwise interactions and cannot directly represent relations involving multiple entities. 
This limitation has motivated the extension of GSP to higher-order relational domains, giving rise to Topological Signal Processing (TSP) over simplicial or cell complexes~\cite{barbarossa2020topological,schaub2021signal,yang2022simplicial,battiloro2024generalized,grimaldi2026topological}.
Here, in addition to nodes, signals are associated to edges, polygons, and higher-dimensional cells. 
Building on this framework, recent work has further developed the relation between TSP and learning over simplicial and cell complexes~\cite{isufi2025topological}.

However, both GSP and TSP retain the fundamental assumption that signals attached to comparable cells belong to a common vector space, which is unrealistic in heterogeneous networks. 
Consider, for example, a common physical phenomenon observed by a multimodal sensing apparatus, composed of a microphone, an accelerometer, and a camera. 
Although the underlying context is shared, the individual signal spaces differ in dimension, coordinates, units, and semantics. 
Consequently, these signals cannot be directly compared or modeled through square matrix-valued weights alone.
Instead, the signal model must explicitly encode how individual signals are transported, observed, projected, or constrained across heterogeneous spaces. 
This calls for a framework in which the geometry of local signal spaces and the maps between them are intrinsic to the signal model itself.

Category theory~\cite{mac1971categories} offers the natural formal language for this problem: 
it is axiomatically centered not on mathematical structures in isolation, but on the relations, or \emph{morphisms}, that hold between them, together with the rules for composing such relations in a consistent manner. 
Network sheaves instantiate this language over a graph~\cite{curry2014sheaves}.
A network sheaf assigns a local data space, called a \emph{stalk}, to each node and edge of a graph, together with \textit{restriction maps} associated with node--edge incidences. 
These maps specify how data from distinct local spaces are transported on a common edge space. 
Network sheaves do not simply enrich the relational domain, as TSP does; they modify the signal model itself  incorporating linear relations among stalks. 
A sheaf signal is simply an assignment of values to the stalks. 
This perspective has recently attracted growing interest in spectral theory, dynamical systems, learning, and neural architectures \cite{ayzenberg2025sheaf}. 
However, its development as a unified SP framework including principled notions of spectral representation, filtering, and sampling for heterogeneous stalk-valued signals remains underinvestigated.

\spara{Related work.} 
Most GSP works restrict signals to scalar or homogeneous vector spaces and encode inter-node relationships through scalar edge weights \cite{ortega2018graph}. 
Several extensions consider vector-valued signals. 
Time-varying graph signals are naturally modeled within time-vertex SP, which jointly exploits graph and temporal structures through product-domain harmonic analysis and filtering \cite{Grassi2018TimeVertex}. 
More generally, multiway GSP and multidimensional graph forecasting extend GSP to tensor-valued data and vector-valued time series by leveraging product-graph constructions across multiple domains \cite{Stanley2020Multiway,Natali2020Forecasting}.
Vector-valued graph trend filtering has also been developed using group-sparsity and mixed-norm regularization \cite{Varma2020VectorValuedGTF}. 
A related direction considers matrix-weighted graphs, where matrix-valued edge weights define Laplacian-like operators acting on vector-valued node states \cite{Trinh2018MatrixWeightedConsensus}. 
Despite their broader modeling capabilities, these approaches still assume that signals belong to a common coordinate space and interact through prescribed linear operators.

The groundwork for applying sheaf-theoretic concepts to SP was initially suggested by Robinson, who recognized that several classical SP problems, such as filtering, sampling, and sensor integration, admit sheaf-theoretic reformulations based on sheaf-morphisms and cohomology~\cite{robinson2012asynchronous,robinson2013understanding,robinson2015sheaf,robinson2014topological,robinson2017sheaves}.
In contrast, our work adopts a model-based SP perspective, generalizing GSP and TSP, in which the network sheaf is regarded as the underlying signal model. 
We exploit this model to systematically develop the fundamental SP tools, i.e., signal representation, spectral analysis, filtering, and sampling, from the sheaf structure. 
In addition, unlike previous works that fix the signal sheaf to a single representation level, our formalism introduces and explicitly builds a \emph{representation sheaf} capturing the most appropriate low-dimensional local representation of the signal, enabling representation-aware SP while remaining fully consistent with the signal sheaf.

A complementary line of work has investigated sheaves as mathematical models for heterogeneous relational data. 
The spectral foundations of cellular sheaves were established by generalizing key notions of spectral graph theory to the sheaf Laplacian~\cite{hansen2019toward}. 
This framework has since been applied to opinion dynamics, sheaf neural networks, and geometric deep learning, where sheaf-based diffusion and convolutions support heterogeneous and asymmetric relations~\cite{hansen2021opinion,hansen2020sheaf,bodnar2022neural,battiloro2024tangent}. 
More recently, these ideas have been extended to infinite-dimensional stalks through Hilbert bundles~\cite{tandon2026consistent}, and have found applications in optimization, semantic communications, and causal abstraction~\cite{ghrist2022cellular,grimaldi2025learning,dacunto2026networkscausalabstractionssheaftheoretic,dacunto2026learningconsistentcausalabstraction}. 

Collectively, these two lines of works demonstrate both the applicability of sheaf theory to concrete signal-processing case studies and the expressive power of sheaves for representing heterogeneous local information and structured relations. 
Nevertheless, neither provides a unified SP framework encompassing signal representation, spectral analysis, filtering, and sampling for heterogeneous stalk-valued signals.

\spara{Contributions.}
To the best of our knowledge, this paper introduces the first unified framework for \emph{sheaf-theoretic signal processing} on networks with heterogeneous local signal spaces. 
Our main contributions are summarized as follows:
\begin{enumerate}
    \item \textit{Signal-model-aware sheaf representations.}
    We start from a network {\it signal sheaf}, describing raw data and their relations, and build a lower dimensional {\it representation sheaf} accommodating local orthogonal bases, dictionaries, latent embeddings, and task-specific feature maps. 
    We show that the consistent mappings between representation and signal sheaves are category-theoretic {\it natural transformations}, and derive the conditions under which they preserve the sheaf structure. 
    This enables SP to be carried out directly in low-dimensional, application-adapted representation spaces rather than on raw observations. 
    Furthermore, we equip representation stalks with metric tensors induced from the raw-signal geometry, providing a unified formulation for heterogeneous signal spaces with different dimensions, physical units, and scales.
    
    \item \textit{Sheaf spectral representation and filtering.}
    Building on the representation sheaf, we develop a spectral framework for heterogeneous sheaf signals. 
    We define the Sheaf Fourier Transform (SFT) from the eigendecomposition of the representation sheaf Laplacian and show that its frequencies jointly capture the network topology, the local geometry, and the transport constraints encoded by the restriction maps. 
    We establish the spectral intertwining between the representation and signal sheaf Laplacians and show that the intrinsic spectral objects are basis-invariant eigenspaces rather than individual eigenvectors. 
    These results naturally lead to polynomial sheaf filters and guarantee interoperability across sheaf representations.

    \item \textit{Sampling and recovery of bandlimited sheaf signals.}
    We introduce a sampling operator that jointly selects network nodes and intra-stalk components. 
    We derive necessary and sufficient conditions for perfect recovery of bandlimited sheaf signals and show that the resulting theory recovers graph-signal sampling as a special case. 
    Importantly, the spectral intertwining induces an orthogonal decomposition of the band into two components, one aligned with the local models, and one orthogonal to it. This offers a more parsimonious and computationally cheaper recovery strategy, by targeting only the part of the band aligned with the signal models. 
    Finally, we formulate sampling-set design as a rank-maximization problem and propose a greedy algorithm with an energy-based tie-breaking criterion.
\end{enumerate}
Numerical experiments on synthetic and real datasets assess the performance of the proposed framework and demonstrate consistent improvements over canonical GSP methods.

%% file: sections/1-primer-on-network-sheaf.tex

\section{A primer on sheaves over graphs}\label{sec:primer}
This section provides a selective overview of the mathematical framework of network sheaves which is key to our work.
The discussion is kept at the general level.
\Cref{sec:structured_signals,sec:Spectral theory over a network sheaf} will specialize the tools for the SP domain. 

\spara{Category theory as the language of relations.}
As anticipated in Section~\ref{sec:introduction}, sheaf-theoretic signal processing (SSP) adopts a relational viewpoint, where the emphasis is placed on the relationships among local mathematical structures rather than on the structures themselves. 
This viewpoint naturally leads to \emph{category theory}, a mathematical framework describing structured objects and the transformations between them \cite{mac1971categories}.

A \emph{category} consists of a collection of objects together with structure-preserving transformations, called \emph{morphisms}, between them. 
Typical examples include the category $\mathbf{Set}$ of sets and functions over sets, the category $\mathbf{Graph}$ of graphs and graph homomorphisms, and the category $\mathbf{Hilb}$ of Hilbert spaces and linear operators. 
Throughout this paper we shall work with the latter, although the concepts introduced below apply to arbitrary target categories.

A graph can itself be viewed as a category, namely a partially-ordered set (poset), whose objects are nodes and edges, and whose morphisms encode the node-edge incidence relations. 
A \emph{functor} is defined as a mapping between categories. 
A \emph{network sheaf} is a functor assigning to every object of the poset an object of a target category, and to every incidence relation a corresponding morphism.
The objects associated with the nodes and edges are called \emph{stalks}, while the associated morphisms are the \emph{restriction maps}. 
In the present work, every stalk is a finite-dimensional Hilbert space and every restriction map is a linear operator encoded as a matrix.

Besides assigning local mathematical structures to the graph, a network sheaf possesses a fundamental gluing property: 
local data that are compatible through the restriction maps uniquely determine a global section over the network, and every global section restricts consistently to the corresponding local data. 

A \emph{natural transformation} is a way of mapping one functor into another functor while completely preserving all structural relationships. 
In this work, we will show how to build a natural transformation from the {\it signal sheaf}, where the raw signals live, and the {\it representation sheaf}, a structure that exploits all possible lower-dimensional representations on each node/edge. 
We will show how the natural transformation allows us to transport fundamental SP operations, like spectral analysis and filtering, from one sheaf to the other, with substantial savings in terms of complexity.

\spara{Network sheaves on graphs valued in \Hilb.}
Let $\graph = (\vertexset, \edgeset)$ be an undirected graph with node set $\vertexset$ and edge set $\edgeset$, where $|\vertexset|=N$ and $|\edgeset|=E$, respectively. 
A \emph{network sheaf} $\ns$ on $\graph$ valued in $\Hilb$ consists of:
\begin{squishenum}
    \item a finite-dimensional real Hilbert space $\ns(i)$ for each node $i \in \vertexset$, called the \emph{node stalk} at $i$;
    \item a finite-dimensional real Hilbert space $\ns(e)$ for each edge $e \in \edgeset$, called the \emph{edge stalk} at $e$;
    \item a bounded linear map $\restrictionmap{i}{e} : \ns(i) \rightarrow \ns(e)$ for each node-edge incident pair $i \trianglelefteq e$, called the \emph{restriction map}.
\end{squishenum}
In the finite-dimensional setting, stalks are Euclidean spaces $\ns(i) \cong \reall^{d_i}$ and $\ns(e) \cong \reall^{d_e}$, and the restriction maps are matrices $\restrictionmap{i}{e} \in \reall^{d_e \times d_i}$. 
Each stalk carries an inner product induced by a \emph{metric tensor} $\G_i \in \pd^{d_i}$ at node $i$ and $\G_e \in \pd^{d_e}$ at edge $e$, so that $\Eprod{\x_i}{\mathbf{y}_i}{\ns(i)} = \x_i^\top \G_i \mathbf{y}_i$, with induced norm $\|\x_i\|_{\G_i}$ (cf. Supp.~S1-A).
The local inner products and the restriction maps play complementary roles: 
the former define the intrinsic geometry of each stalk, whereas the latter specify how data attached to adjacent nodes are compared through their common edge space, in its intrinsic geometry. Thus, heterogeneous stalks with different units, scales, or noise levels are handled without any ad hoc rescaling.  

The dimensions $d_i$ and $d_e$ are in general independent. Usually $d_e \geq \max(d_i, d_j)$, for an edge $e=(i,j)$. 
In this case, the edge stalk acts as an \emph{integration space}: 
the restriction maps embed the node vectors into a richer common space, and consistency demands a  global agreement between $\x_i$ and $\x_j$. 
However, we can also have $d_e \leq \min(d_i, d_j)$. In this case, the edge stalk acts as a \emph{compression bottleneck} and consistency requires that the vectors $\x_i$ and $\x_j$ agree on a common low-dimensional projection, allowing partial agreement while tolerating discrepancies in the complementary directions. 

GSP can be seen as the special case $\ns(i) \cong \reall^d$, $\ns(e) \cong \reall^d$, and $\restrictionmap{i}{e} = \eye{d}$, with typically $d=1$ and $\G_e = w_e \eye{d}$ for all $i \trianglelefteq e$, where $w_e \in \reall_{+}$ is the edge weight.

\spara{Cochains, sections and consistency.}
The space of \emph{node vectors} ($0$-cochains) is the direct sum of the node stalks, $C^0(\graph; \ns) = \bigoplus_{i \in \vertexset} \ns(i)$, and a node vector is $\x = [\x_1^\top,\ldots, \x_i^\top, \ldots,\x_N^ \top]^\top \in \reall^{D}$ with $D=\sum_i d_i$, i.e. the concatenation of all node stalk vectors. 
A \emph{section} of $\ns$ over a subgraph $\subgraph \subseteq \graph$ is a node vector that is \emph{locally consistent} on every edge internal to \subgraph:
\begin{equation}\label{eq:local_consistency}
    \restrictionmap{i}{e} \x_i = \restrictionmap{j}{e} \x_j\,, \qquad \forall \, e \in \mathcal{U}\,.
\end{equation}
A \emph{global section} is an assignment of a signal to every stalk of the sheaf such that all restriction maps are simultaneously satisfied. 
The space of global sections is denoted $\gsspace{\graph}{\ns}$. 
The hierarchy is therefore: 
$0$-cochains (no constraint) $\supset$ sections over \subgraph (local consistency) $\supset$ global sections (full consistency).

Consistency is a \emph{structural} condition, not a regularity one. 
In GSP, smoothness quantifies how much a vector varies across edges. 
Here, local consistency means that the views of the vector at adjacent nodes are \emph{exactly compatible} through the restriction maps. 
A global section lies in the kernel of the sheaf Laplacian, as we show next.

The space of \emph{edge vectors} ($1$-cochains) is $C^1(\graph; \ns) = \bigoplus_{e \in \edgeset} \ns(e)$, with elements $\x=[\x_{e_1}^\top,\ldots,\x_{e_E}^\top]^\top \in \reall^{\sum_e d_e}$.
Edge vectors generalize the 1-cochains of TSP: 
TSP corresponds to $\ns(e) = \reall$ with trivial restriction maps, while here each edge carries a vector in its own space $\ns(e)$. 
The metric tensors assemble into block-diagonal matrices $\mathbb{G}_0 \in \pd^{\sum_i d_i}$ and $\mathbb{G}_1 \in \pd^{\sum_e d_e}$, inducing inner products on $C^0$ and $C^1$.

\spara{Total variation and the sheaf Laplacian.}
As with GSP, a fundamental descriptor of a graph is its incidence matrix. 
Given an orientation $e=(i,j)$ (head $i$, tail $j$), the \emph{coboundary operator on a sheaf} $\B : C^0(\graph;\ns)  \rightarrow C^1(\graph;\ns)$ maps a node vector to an edge vector that measures the local inconsistency:
\begin{equation}\label{eq:coboundary}
    (\B\x)_{e} = \restrictionmap{i}{e}\x_i - \restrictionmap{j}{e}\x_j\,.
\end{equation}
\B is a block matrix with blocks
\begin{equation}\label{eq:coboundary_row}
    \B(e,i) \coloneq \begin{cases}
        \restrictionmap{i}{e} & \text{if $i$ is head,}\\
        -\restrictionmap{i}{e} & \text{if $i$ is tail,}\\
        \zeros_{d_e \times d_i} & \text{otherwise.}
    \end{cases}
\end{equation}

The \emph{total variation} of a vector $\x$ is defined as:
\begin{equation}\label{eq:total_variation}
    \mathrm{TV}(\x) = \|\B\x\|_{\mathbb{G}_1}^2 = \sum_{e \in \edgeset} \|\restrictionmap{i}{e}\x_i - \restrictionmap{j}{e}\x_j\|_{\G_{e}}^2\,;
\end{equation}
which reduces to $\sum_{e\in\edgeset}w_{e} (x_i - x_j)^2$ in the GSP case. 
Introducing the adjoint of $\B$ as $\B^\star = (\mathbb{G}_0)^{-1}\B^\top\mathbb{G}_1$ (cf. Supp.~S1-B), the \emph{sheaf Laplacian} is
\begin{equation}\label{eq:sheaf_laplacian}
    \mathbf{L}_\ns = \B^\star \B\,,
\end{equation}
so that $\mathrm{TV}(\x) = \Eprod{\x}{\mathbf{L}_\ns \x}{C^0}$. 
The operator $\mathbf{L}_\ns$ is symmetric positive semidefinite on $C^0(\graph;\ns)$ by construction (cf. Supp.~S1-B), and $\ker(\mathbf{L}_\ns) = \gsspace{\graph}{\ns}$.
The sheaf Laplacian depends on both the restriction maps and the metric tensors $\G_e$ through the adjoint $\B^\star$. 
Since $\ker(\mathbf{L}_\ns) = \ker(\B)$, the kernel depends only on the restriction maps and is unaffected by the choice of the metric. 
The metric tensors instead shape the non-zero spectrum: 
with a general anisotropic $\G_e$, the non-zero eigenvectors depend on the interplay between the restriction maps and the metric.
The dimension of the kernel, and hence the number of zero eigenvalues, depends jointly on the graph topology and the restriction maps.

\begin{remark}[Block structure of the sheaf Laplacian]
    The sheaf Laplacian $\mathbf{L}_\ns \in \reall^{D \times D}$ has a block structure inherited from the direct sum decomposition of $C^0(\graph;\ns)$. 
    The $(i,j)$-th block of size $d_i \times d_j$ is
    \begin{equation}
        [\mathbf{L}_\ns]_{ij} = \begin{cases}
            \G_i^{-1} \sum_{e \in \starr{i}} \restrictionmap{i}{e}^\top \G_e \, 
            \restrictionmap{i}{e} & \text{if } i = j, \\
            -\G_i^{-1} \restrictionmap{i}{e}^\top \G_{e} \, 
            \restrictionmap{j}{e} & \text{if } e \in \edgeset, \\
            \zeros_{d_i \times d_j} & \text{otherwise;}
        \end{cases}
    \end{equation}
    where $\starr{i} = \{e \in \edgeset \mid i \trianglelefteq e\}$ denotes the set of edges incident to $i$.
    The diagonal blocks are positive semidefinite and encode the local geometry and the restriction maps at each node; the off-diagonal blocks encode the coupling between adjacent nodes through the edge metric. 
    In the scalar GSP case, $d_i = 1$, $\restrictionmap{i}{e} = 1$, and $\G_e = w_e$, so the off-diagonal block reduces to $-w_{e}$ and $\mathbf{L}_\ns$ recovers the weighted combinatorial Laplacian.
\end{remark}

Note that $\mathbf{L}_\ns$ is \emph{not} symmetric in the standard sense.
Rather, $\mathbf{L}_\ns$ is $\mathbb{G}_0$-\emph{selfadjoint} (cf. Supp.~S1-B), i.e., 
$\Eprod{\x}{\mathbf{L}_\ns\mathbf{y}}{C^0} = \Eprod{\mathbf{L}_\ns\x}{\mathbf{y}}{C^0}$ for all $\x,\mathbf{y}\in C^0(\graph;\ns)$.
The symmetrized version $\widetilde{\mathbf{L}}_\ns = \mathbb{G}_0^{1/2}\mathbf{L}_\ns\mathbb{G}_0^{-1/2}$ is symmetric in the standard sense and shares the same spectrum as $\mathbf{L}_\ns$.

\begin{remark}[Unweighted case]
    When $\mathbb{G}_0 = \eye{}$ and $\mathbb{G}_1 = \eye{}$, we have $\B^\star = \B^\top$ and $\mathbf{L}_\ns = \B^\top\B$.
\end{remark}

\begin{remark}[Duality: cosheaves and extension maps]
    Every network sheaf $\ns$ has a dual structure, the \emph{network cosheaf}, obtained by reversing the direction of the restriction maps: 
    the \emph{extension map} on an incident pair $i \trianglelefteq e$ is the adjoint $\extensionmap{i}{e} = \restrictionmap{i}{e}^\top : \ns(e) \rightarrow \ns(i)$. 
    The cosheaf governs the boundary operator $\B^\star$ and the Laplacian on edge vectors. 
    Refer to~\cite{curry2014sheaves} for a comprehensive discussion.
\end{remark}

%% file: sections/2-structured-sheaf-signals.tex

\section{From Signal Sheaf to Representation Sheaf via Natural Transformations}
\label{sec:structured_signals}
When the stalks of the network sheaf \ns in \cref{sec:primer} correspond to observed signal spaces, and the restriction maps correspond to transport operators among these spaces, we term \ns a \emph{signal sheaf}.
In this work, we assume that \ns is either given or constructed from prior knowledge, leaving its data-driven estimation for future work.
In many SP applications, the vectors associated with the node stalks admit parsimonious representations over suitable dictionaries. 
Rather than processing the original high-dimensional signals, it is therefore desirable to operate directly on their low-dimensional representation vectors. 
This naturally raises the following question: 
\emph{under which conditions does a collection of local dictionary representations induce a consistent low-dimensional network sheaf?}

Specifically, let the signal attached to node $i$ admit the representation
\begin{equation}
\label{eq:dictionary_representation}
\x_i=\D_i\vecs_i,
\qquad
\vecs_i\in\reall^{c_i},
\qquad
c_i\le d_i,
\end{equation}
where $\D_i\in\reall^{d_i\times c_i}$ is a dictionary associated with node $i$. 
Without loss of generality, we assume that the dictionaries belong to the Stiefel manifold,
\begin{equation}
\label{eq:stiefel_manifold}
\stiefel{n}{m}
=
\left\{
\V\in\reall^{n\times m}
\,\middle|\,
\V^\top\V=\eye{m},
\;
m<n
\right\}.
\end{equation}

The representation vector $\vecs = [\vecs_1^\top, \ldots,\vecs_N^\top]^\top \in \reall^{C}$ with $C=\sum_i c_i$, naturally define the node stalks of a \textit{representation sheaf}, denoted by $\nscoeff$, with
\begin{equation}
    \nscoeff(i)\cong\reall^{c_i}\,, \qquad i\in\vertexset.    
\end{equation}
Moreover, the metric tensor $\G_i$ of the original stalk induces the metric
\begin{equation}
\label{eq:induced_metric_node}
\Hmet_i = \D_i^\top\G_i\D_i \in\pd^{c_i}\,,
\end{equation}
and similarly, for each edge stalk,
\begin{equation}\label{eq:induced_metric_edge}
    \Hmet_e = \D_e^\top\G_e\D_e\,,    
\end{equation}
where $\D_e$ denotes the dictionary over the edge $e$.
These local metrics assemble into the block-diagonal matrices $\mathbb{H}_0$ and $\mathbb{H}_1$, which replace $\mathbb{G}_0$ and $\mathbb{G}_1$ in the representation sheaf.

\begin{figure}[!t]
\centering
\resizebox{\columnwidth}{!}{%
\begin{tikzpicture}
\node[draw, circle, minimum size=1.0cm, inner sep=1pt,
      font=\large] (ni) at (0, 0) {$i$};

\node[draw, circle, minimum size=1.0cm, inner sep=1pt,
      font=\large] (nj) at (7.5, 0) {$j$};

\node[font=\large] (elbl) at (3.75, 0.30) {$e$};
\draw[thick] (3.25, 0) -- (4.25, 0);
\coordinate (etop) at (3.75, 0);

\draw[->, shorten >=20pt, shorten <=5pt]
    (ni.east) -- node[above, font=\large] {$i \trianglelefteq e$} (etop);
\draw[->, shorten >=20pt, shorten <=5pt]
    (nj.west) -- node[above, font=\large] {$j \trianglelefteq e$} (etop);

\node[font=\large, text=mylightblue] (Ji) at (0,    -2.2) {$\nscoeff(i)$};
\node[font=\large, text=mylightblue] (Je) at (3.75, -2.2) {$\nscoeff(e)$};
\node[font=\large, text=mylightblue] (Jj) at (7.5,  -2.2) {$\nscoeff(j)$};

\draw[->]
    (Ji) -- node[above, font=\large, text=mylightblue]
    {$\restrictionmapscoeff{i}{e}$} (Je);
\draw[->]
    (Jj) -- node[above, font=\large, text=mylightblue]
    {$\restrictionmapscoeff{j}{e}$} (Je);

\node[font=\large, text=mydarkviolet] (Fi) at (0,    -4.5) {$\ns(i)$};
\node[font=\large, text=mydarkviolet] (Fe) at (3.75, -4.5) {$\ns(e)$};
\node[font=\large, text=mydarkviolet] (Fj) at (7.5,  -4.5) {$\ns(j)$};

\draw[->]
    (Fi) -- node[below, font=\large, text=mydarkviolet]
    {$\restrictionmap{i}{e}$} (Fe);
\draw[->]
    (Fj) -- node[below, font=\large, text=mydarkviolet]
    {$\restrictionmap{j}{e}$} (Fe);

\draw[->, myorange] (Ji) -- node[left,  font=\large, text=myorange] {$\D_i$} (Fi);
\draw[->, myorange] (Je) -- node[right, font=\large, text=myorange] {$\D_e$} (Fe);
\draw[->, myorange] (Jj) -- node[right, font=\large, text=myorange] {$\D_j$} (Fj);

\draw[->, dashed, mylightblue]
    (ni.south) to[out=-110, in=110] (Ji.north);
\draw[->, dashed, mydarkviolet]
    (ni.south) to[out=-130, in=150] (Fi.north west);

\draw[->, dashed, mylightblue]
    (nj.south) to[out=-70, in=70] (Jj.north);
\draw[->, dashed, mydarkviolet]
    (nj.south) to[out=-50, in=30] (Fj.north east);

\draw[->, dashed, mylightblue]
    (etop) to[out=-70, in=70] (Je.north east);
\draw[->, dashed, mydarkviolet]
    (etop) to[out=-110, in=110] (Fe.north west);

\draw[->, dashed, mylightblue]
    ($(ni)!0.5!(etop)+(0,-0.2)$) to[out=-110, in=110]
    ($(Ji)!0.5!(Je)+(0,0.5)$);
\draw[->, dashed, mylightblue]
    ($(nj)!0.5!(etop)+(0,-0.2)$) to[out=-70, in=70]
    ($(Jj)!0.5!(Je)+(0,0.5)$);

\draw[->, dashed, mydarkviolet]
    ($(ni)!0.5!(etop)+(0,-0.2)$) to[out=-130, in=130]
    ($(Fi)!0.5!(Fe)+(0,0.2)$);
\draw[->, dashed, mydarkviolet]
    ($(nj)!0.5!(etop)+(0,-0.2)$) to[out=-50, in=50]
    ($(Fj)!0.5!(Fe)+(0,0.2)$);
 
\end{tikzpicture}%
}
\caption{Dictionaries as components of the \myorange{natural transformation $\delta$}~\cite{mac1971categories} between the \mylightblue{representation sheaf \nscoeff} and the \mydarkviolet{signal sheaf \ns} over the edge $e=(i,j)$.}
\label{fig:natural_transformations}
\end{figure}
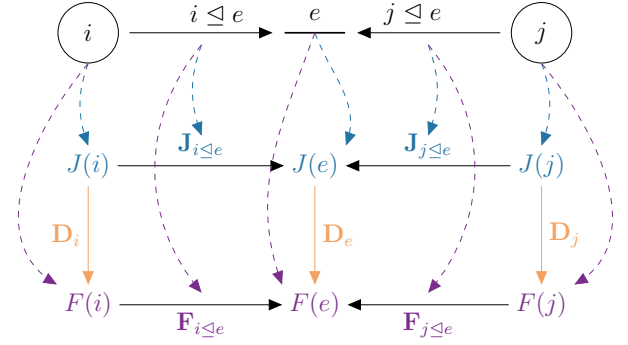

Unlike the node dictionaries, however, the edge dictionaries $\{\D_e\}_{e\in\edgeset}$ and the representation restriction maps $\restrictionmapscoeff{i}{e}$ are not known a priori.
Their existence is constrained by the requirement that the representation preserves the local consistency relations of the signal sheaf.
This requirement is expressed by the commutative diagrams illustrated in Fig.~\ref{fig:natural_transformations}, namely
\begin{equation}
\label{eq:naturality}
\restrictionmap{i}{e}\D_i
=
\D_e\restrictionmapscoeff{i}{e},
\qquad
\restrictionmap{j}{e}\D_j
=
\D_e\restrictionmapscoeff{j}{e},
\end{equation}
for every edge $e=(i,j)$.

Equation~\eqref{eq:naturality} naturally leads to the following question: given the node dictionaries $\{\D_i\}_{i\in\vertexset}$, under which conditions does there exist a consistent representation sheaf $\nscoeff$?
The next proposition provides the answer.

\begin{proposition}[Existence of local representation maps] \label{prop:local_representation_maps_existence} 
Let $e=(i,j)$ be an edge of $\graph$, and assume that the node dictionaries $\D_i\in\stiefel{d_i}{c_i}$ and $\D_j\in\stiefel{d_j}{c_j}$ are given. 
Define 
\begin{equation} \label{eq:edge_lifted_concat} 
\A_e \coloneqq \begin{bmatrix} \restrictionmap{i}{e}\D_i & \restrictionmap{j}{e}\D_j \end{bmatrix}. 
\end{equation} 
Then there exist an edge dictionary $\D_e\in\stiefel{d_e}{c_e}$ and representation restriction maps $\restrictionmapscoeff{i}{e}$, $\restrictionmapscoeff{j}{e}$ satisfying \eqref{eq:naturality} if and only if 
\begin{equation} 
\label{eq:rank_condition_representation_maps} 
\mathrm{rank}(\A_e)\leq c_e. 
\end{equation} 
Equivalently, the images of $\restrictionmap{i}{e}\D_i$ and $\restrictionmap{j}{e}\D_j$ must lie in a common subspace of $\ns(e)$ of dimension at most $c_e$. 
\end{proposition} 
\begin{proof} 
See Supp. S2. 
\end{proof}

Having established the necessary and sufficient condition for the existence of a consistent local representation sheaf, we now provide an explicit construction of the corresponding edge dictionaries and consistent representation restriction maps.

\begin{proposition}[Construction of local representation maps]
\label{prop:local_representation_maps_construction}
Under the assumptions of \Cref{prop:local_representation_maps_existence}, suppose that $\mathrm{rank}(\A_e)\leq c_e$. 
Let $r_e=\mathrm{rank}(\A_e)$, and let $\U_e\in\reall^{d_e\times r_e}$ be an orthonormal basis of $\mathrm{Im}(\A_e)$.
Complete $\U_e$ to a matrix $\D_e\in\stiefel{d_e}{c_e}$ by adding $c_e-r_e$ orthonormal columns.
Then the representation restriction maps can be chosen as
\begin{equation}
\label{eq:constructive_representation_restrictions}
\restrictionmapscoeff{i}{e}
=
\D_e^\top
\restrictionmap{i}{e}
\D_i,
\qquad
\restrictionmapscoeff{j}{e}
=
\D_e^\top
\restrictionmap{j}{e}
\D_j.
\end{equation}
With this choice, the commutativity conditions \eqref{eq:naturality} hold true.
\end{proposition}
\begin{proof}
    See Supp. S2.
\end{proof}
The previous propositions provide {\it local} conditions for constructing the representation over a single edge. 
If these conditions hold for every edge, the collection of dictionaries $\delta=\{\D_o\}_{o \in \vertexset \cup \edgeset}$ represents a \emph{natural transformation}~\cite{mac1971categories} from the representation sheaf $\nscoeff$ to the signal sheaf $\ns$ 
\begin{equation}\label{eq:delta}
    \delta: \nscoeff \Rightarrow \ns\,.
\end{equation}
The identification of a natural transformation $\delta$ from $\nscoeff$ to $\ns$ is useful not only to reconstruct $\x$ from $\vecs$, but also to transport consistency properties from one sheaf to the other, as it will be extensively shown in the following.

Let us represent the node and edge components of $\delta$ as block-diagonal matrices
\begin{equation}\label{eq:delta_components}
    \mathbb D_0 \coloneqq \mathrm{blkdiag}(\{\D_i\}_{i\in\vertexset})\,,
    \;\;
    \mathbb D_1 \coloneqq \mathrm{blkdiag}(\{\D_e\}_{e\in\edgeset})\,,
\end{equation}
acting respectively on $0$- and $1$-cochains, $\mathbb D_0:\cochainspace{0}{\graph}{\nscoeff}\to\cochainspace{0}{\graph}{\ns}$ and $\mathbb D_1:\cochainspace{1}{\graph}{\nscoeff}\to\cochainspace{1}{\graph}{\ns}$.
Then, it is not difficult to prove the following.
\begin{proposition}[Global section correspondence]\label{prop:section_correspondence}
    Let $\delta$ as in \cref{eq:delta} be a natural transformation with node component $\mathbb D_0$.
    Then:
    \begin{enumerate}
        \item[(i)] If $\vecs \in \gsspace{\graph}{\nscoeff}$, then $\x=\mathbb D_0 \vecs$ satisfies $\x \in \gsspace{\graph}{\ns}$.
        \item[(ii)] Conversely, if $\x \in \gsspace{\graph}{\ns}$ and $\x \in \im{\mathbb D_0}$, then $\vecs=\mathbb D_0^\top \x$ satisfies $\vecs \in \gsspace{\graph}{\nscoeff}$ and $\x = \mathbb D_0\vecs$.
    \end{enumerate}
    Consequently, $\mathbb D_0$ restricts to a linear isomorphism
    \begin{equation}\label{eq:section_iso}
        \gsspace{\graph}{\nscoeff} \;\cong\; \gsspace{\graph}{\ns} \cap \im{\mathbb D_0}\,,
        \quad \im{\mathbb D_0} \coloneqq \bigoplus_{i \in \vertexset}\im{\D_i}\,.
    \end{equation}
\end{proposition}
\begin{proof}
    See Supp. S2.
\end{proof}

\begin{remark}[Interpretation of image containment in SP]\label{rem:image_containment}
    The image-containment condition $\im{\A_e}\subseteq\im{\D_e}$ in \cref{prop:local_representation_maps_existence} has a natural interpretation in SP: 
    the edge dictionary $\D_e$ must be expressive enough to represent any signal that are transported to the edge stalk  from node $i$ and node $j$ through $\restrictionmap{i}{e}$ and $\restrictionmap{j}{e}$. 
    In other words, $\restrictionmap{i}{e}$ and $\restrictionmap{j}{e}$ should not map the node signal models $\im{\D_i}$ and $\im{\D_j}$ outside the edge signal model $\im{\D_e}$.
\end{remark}

%% file: sections/3-sheaf-spectral-theory.tex
\section{Spectral theory over a network sheaf}\label{sec:Spectral theory over a network sheaf}
In this section, we introduce the Sheaf Fourier Transform (SFT) as a generalization of the GFT. 
Since the signal sheaf $\ns$ and the representation sheaf $\nsrep$ are related through the natural transformation $\delta$ in \cref{eq:delta}, the SFT can be defined on either sheaf. 
Unless otherwise specified, we develop the spectral theory on the representation sheaf, whose lower-dimensional stalks provide a more compact representation while preserving the spectral structure of the signal sheaf (cf. \cref{Spectral correspondence between the signal and representation sheaves}). 

\subsection{Sheaf Fourier Transform}
Consider the representation sheaf $\nsrep$ over the graph $\graph$, where $|\vertexset|=N$ and $|\edgeset|=E$.
Generalizing the approach used in GSP, the spectral analysis of sheaf signals is based on the eigendecomposition of the sheaf Laplacian with respect to the inner product induced by the metric tensor $\mathbb{H}_0$.
Since $\mathbf{L}_{\nsrep}$ is self-adjoint with respect to $\Eprod{\cdot}{\cdot}{C^0}$, the corresponding Euclidean symmetric operator (cf. Supp.~S1-B) is
\begin{equation}
\label{eq:symmetrized_laplacian}
\widetilde{\mathbf{L}}_{\nsrep}
=
\mathbb{H}_0^{1/2}
\mathbf{L}_{\nsrep}
\mathbb{H}_0^{-1/2}
=
\widetilde{\U}
\boldsymbol{\Lambda}
\widetilde{\U}^{\top},
\qquad
\widetilde{\U}^{\top}\widetilde{\U}=\eye{}.
\end{equation}

Defining
\begin{equation}
    \U=\mathbb{H}_0^{-1/2}\widetilde{\U}\,,    
\end{equation}
we obtain the generalized eigendecomposition
\begin{equation}
\label{eq:sheaf_spectral}
\mathbf{L}_{\nsrep}\U
=
\U\boldsymbol{\Lambda},
\qquad
\U^{\top}\mathbb{H}_0\U=\eye{},
\end{equation}
where $\boldsymbol{\Lambda} = \mathrm{diag}(\lambda_1,\ldots,\lambda_C)$, $0\le\lambda_1\le\cdots\le\lambda_C$.
The columns $\{\vecu_\ell\}_{\ell=1}^{C}$ of $\U$ are called the \emph{sheaf Fourier modes}, while the corresponding eigenvalues $\{\lambda_\ell\}$ are the \emph{sheaf frequencies}.
Each mode $\vecu_\ell=[(\vecu_\ell^1)^\top,\ldots,(\vecu_\ell^N)^\top]^\top$ contains one block $\vecu_\ell^i$ for every node stalk.

\begin{definition}[Sheaf Fourier Transform]
\label{def:sft}
The Sheaf Fourier Transform (SFT) of $\vecs\in C^0(\graph;\nsrep)$ is the vector of coefficients
\begin{equation}
\label{eq:sft}
\widehat{\vecs} = \U^{\top}\mathbb{H}_0\vecs\,.
\end{equation}
The $\ell$-th Fourier coefficient is
\begin{equation}
\label{eq:sft_coefficient}
\widehat{\vecs}(\lambda_\ell) = \Eprod{\vecs}{\vecu_\ell}{C^0} = \sum_{i\in\vertexset} \vecs_i^{\top} \Hmet_i \vecu_\ell^i \,.
\end{equation}
The inverse SFT is
\begin{equation}
\label{eq:isft}
\vecs = \sum_{\ell=1}^{L} \widehat{\vecs}(\lambda_\ell)\,\vecu_\ell =\U\widehat{\vecs}\,.
\end{equation}
\end{definition}

The SFT computes the coefficients of the expansion of a sheaf signal over the orthonormal basis formed by the eigenvectors of the sheaf Laplacian.
In complete analogy with GSP, the eigenvalues of the Laplacian play the role of frequencies.
However, unlike the graph Laplacian, whose spectrum reflects only the graph topology, the spectrum of the sheaf Laplacian depends {\it jointly} on the graph topology, the restriction maps, and the local metric tensors.
In particular, the sheaf frequency $\lambda_\ell$ measures the inconsistency of the corresponding Fourier mode over the whole sheaf:
\begin{equation}
\lambda_\ell = \mathrm{TV}(\vecu_\ell) = \sum_{e=(i,j)\in\edgeset}\|\restrictionmapscoeff{i}{e}\vecu_\ell^i-\restrictionmapscoeff{j}{e}\vecu_\ell^j\|_{\Hmet_e}^{2}\,.
\end{equation}
From this perspective, the sheaf Fourier modes admit a variational interpretation. 
They constitute the sequence of $\mathbb{H}_0$-orthonormal sheaf signals of progressively increasing total inconsistency: 
the first mode minimizes the total inconsistency among all unit-norm signals, while each subsequent mode minimizes it subject to being orthogonal to all preceding modes. 
This interpretation extends the classical variational characterization of the GFT to network sheaves: 
while graph Fourier modes minimize variation over the graph topology, sheaf Fourier modes minimize inconsistency jointly induced by the graph topology, the restriction maps, and the local Hilbert-space geometry.

\subsection{Spectral correspondence between the signal and representation sheaves}
\label{Spectral correspondence between the signal and representation sheaves}
The SFT introduced in the previous subsection can be defined either on the signal sheaf $\ns$ or on the representation sheaf $\nsrep$. Owing to the natural transformation $\delta$ relating the two sheaves, however, the two spectral domains are not independent. 
As shown next, the corresponding sheaf Laplacians are intertwined by the natural transformation, implying that their spectral decompositions are in one-to-one correspondence within the representation subspace. 
Consequently, the SFT computed on the representation sheaf faithfully captures the spectral properties of the original signal sheaf while operating on lower-dimensional, signal-model-aware, stalks.
\begin{theorem}[Intertwining of the sheaf Laplacians]
\label{th:intertwining}
Let $\delta$ as in \cref{eq:delta} with node component $\mathbb{D}_0$ and edge component $\mathbb{D}_1$ as in \Cref{eq:delta_components}, respectively.
Assume that the signal sheaf and the representation sheaf are endowed with metric tensors
$\{\G_i,\G_e\}$ and $\{\Hmet_i,\Hmet_e\}$ in \cref{eq:induced_metric_node,eq:induced_metric_edge}.
Furthermore, consider that the restriction maps of the signal sheaf are the metric-compatible lifts of those of the representation sheaf:
\begin{equation}
\label{eq:metric_lift}
\restrictionmap{i}{e} = \D_e\,\restrictionmapscoeff{i}{e}\,\Hmet_i^{-1}\D_i^\top\G_i,\qquad \forall\, i\trianglelefteq e .
\end{equation}

Then the corresponding sheaf Laplacians satisfy the intertwining relation
\begin{equation}
\label{eq:spectral_intertwining}
 \mathbf{L}_\ns\,\, \mathbb{D}_0 = \mathbb{D}_0\,\, \mathbf{L}_\nscoeff\,.
\end{equation}
\end{theorem}
\begin{proof}
    See Supp. S2.
\end{proof}

This proposition has an immediate consequence on the spectral theory, as established next.
\begin{corollary}[Spectral correspondence]
\label{cor:spectral_correspondence}
Let $\mathbf{L}_{\nsrep}\vecv=\mu\vecv\,$.
Then
\begin{equation}
    \mathbf{L}_{\ns}\mathbb{D}_0\vecv=\mu\,\mathbb{D}_0\vecv\,.    
\end{equation}
Hence, every eigenvalue of the representation sheaf Laplacian is an eigenvalue of the signal sheaf Laplacian, and the natural transformation maps every eigenspace of $\mathbf{L}_{\nsrep}$ into the eigenspace of $\mathbf{L}_{\ns}$ associated with the same eigenvalue.

Additionally, since  
\begin{equation}
\label{isometry}
\left\langle
\mathbb{D}_0\vecv,\mathbb{D}_0\vecw
\right\rangle_{\mathcal{C}^{0}(\mathcal{G};F)} =\left\langle \vecv, \vecw \right\rangle_{\mathcal{C}^{0}(\mathcal{G};J)} \,,
\end{equation}
every $\mathbb H_0$-orthonormal eigenbasis of $\mathbf{L}_{\nsrep}$ is mapped into a $\mathbb G_0$-orthonormal set of eigenvectors of $\mathbf{L}_{\ns}$.
\end{corollary}
\begin{proof}
See Supp. S2.
\end{proof}

The previous corollary shows that the representation sheaf is not merely a low-dimensional description of the signal sheaf, but a \emph{spectrally faithful representation} of it. 
The natural transformation preserves the eigenvalues and maps the Fourier modes of the representation sheaf into Fourier modes of the signal sheaf with identical frequencies. 
Consequently, all spectral quantities associated with the representation sheaf admit an equivalent interpretation on the signal sheaf. 
In the remainder of the paper, we therefore develop the spectral theory on the representation sheaf, without loss of information within the subspace generated by the dictionaries.

\subsection{Intrinsic spectral representation}
Unlike graph Laplacians, whose eigenvalues are generically simple, the spectrum of a sheaf Laplacian often exhibits eigenvalues with multiplicity greater than one. 
This phenomenon is not accidental, but reflects the geometric structure encoded by the restriction maps. 
Whenever the latter transport information coherently across the sheaf, the associated eigenspaces acquire additional degrees of freedom, leading to repeated eigenvalues. 
Supp. S3 illustrates two representative mechanisms generating spectral multiplicity, namely the connection graph \cite{chung2014ranking} and the hierarchical Stiefel embedding \cite{dacunto2026networkscausalabstractionssheaftheoretic}.

Importantly, spectral multiplicity makes the individual SFT coefficients depending on the particular orthonormal basis selected within each eigenspace.
Consequently, the Fourier coefficients associated with individual eigenvectors are not intrinsic properties of either the signal or the sheaf.

The intrinsic spectral objects are instead the \emph{eigenspaces} of the sheaf Laplacian.
Let $\{\mu_k\}_{k=1}^{K}$ denote the set of distinct eigenvalues of $\mathbf{L}_{\nsrep}$.
The eigenspace associated with the frequency $\mu_k$ is
\begin{equation}
\label{eq:eigenspace}
\mathcal{E}_k = \ker\!\left(\mathbf{L}_{\nsrep}-\mu_k\eye{C}\right)\,,
\end{equation}
with dimension $m_k$. Let $\U_k\in\reall^{C\times m_k}$ be any $\mathbb{H}_0$-orthonormal basis of $\mathcal{E}_k$, namely
    $\U_k^\top\mathbb{H}_0\U_k=\eye{m_k}$.    
Although the basis $\U_k$ is not unique, the orthogonal projector onto $\mathcal{E}_k$,
\begin{equation}
\label{eq:projector}
\mathbf{P}_k = \U_k\U_k^\top\mathbb{H}_0,
\end{equation}
is uniquely determined by the eigenspace and is therefore independent of the particular basis. Then, the intrinsic spectral component of $\vecs\in C^0(\graph;\nsrep)$ at frequency $\mu_k$ is its orthogonal projection onto $\mathcal{E}_k$,
\begin{equation}
\label{eq:intrinsic_component}
\vecs_k = \mathbf{P}_k\vecs \in \mathcal{E}_k\,.
\end{equation}
Since the eigenspaces are mutually $\mathbb{H}_0$-orthogonal and span $C^0(\graph;\nsrep)$, the projectors satisfy
\begin{equation}
\label{eq:projector_properties}
\mathbf{P}_k\mathbf{P}_\ell = \delta_{k\ell}\mathbf{P}_k\,,
\quad \sum_{k=1}^{K}\mathbf{P}_k = \eye{L}\,, 
\quad \mathbf{L}_{\nsrep} = \sum_{k=1}^{K} \mu_k\mathbf{P}_k\,,
\end{equation}
which is the spectral decomposition of the sheaf Laplacian. The signal then admits the intrinsic spectral decomposition
\begin{equation}
\label{eq:intrinsic_decomposition}
\vecs = \sum_{k=1}^{K}\vecs_k,
\end{equation}
where each component $\vecs_k$ in \eqref{eq:intrinsic_component} contains the portion of the signal associated with the frequency $\mu_k$.
Unlike the classical SFT coefficients, the decomposition in \cref{eq:intrinsic_decomposition} is invariant with respect to any orthogonal change of basis within the eigenspaces. 
Furthermore, the energy carried by the $k$-th spectral component is
\begin{equation}
\label{eq:intrinsic_energy}
\|\vecs_k\|_{\mathbb{H}_0}^{2} = \vecs^\top \mathbb{H}_0 \mathbf{P}_k \vecs,
\end{equation}
and Parseval's identity becomes
\begin{equation}
\label{eq:intrinsic_parseval}
\|\vecs\|_{\mathbb{H}_0}^{2} = \sum_{k=1}^{K} \|\vecs_k\|_{\mathbb{H}_0}^{2}.
\end{equation}
In the absence of repeated eigenvalues, every eigenspace is one-dimensional and the proposed formulation therefore constitutes a natural extension of the GFT to network sheaves: 
it preserves the classical spectral representation whenever the spectrum is simple, while remaining well-defined and invariant under orthogonal changes of basis within the eigenspaces when repeated eigenvalues occur.

%% file: sections/4-sheaf_filter.tex
\section{Sheaf filtering}\label{sec:filtering}

Filtering is a fundamental operation in SP. 
In GSP, linear shift-invariant filters are functions of the graph Laplacian, whose spectrum defines the graph frequency domain. 
This framework extends naturally to network sheaves by replacing the graph Laplacian with the sheaf Laplacian. 
Thanks to the spectral correspondence established previously, filtering can be developed directly on the representation sheaf, yielding lower-dimensional filters that preserve the spectral behavior of their counterparts on the signal sheaf.

A linear sheaf filter is any matrix function of the sheaf Laplacian. Among the possible choices, polynomial filters are particularly appealing, since they provide a parsimonious parameterization requiring only a small number of coefficients while preserving locality on the graph. 
Accordingly, we consider the class of polynomial filters
\begin{equation}
\label{eq:poly_filter}
\vecy
=
\sum_{q=0}^{Q}
a_q
\mathbf{L}_{\nsrep}^{\,q}
\vecs,
\end{equation}
where
$\{a_q\}_{q=0}^{Q}$
are the filter coefficients. 
Using the spectral decomposition of $\mathbf{L}_\nscoeff$ in \cref{eq:projector_properties},
the polynomial filter (\ref{eq:poly_filter}) can be rewritten as
\begin{equation}
\label{eq:spectral_filter}
\vecy
=
\sum_{k=1}^{K}
h(\mu_k)\,
\mathbf{P}_k
\vecs \qquad \text{with}\qquad h(\mu)=\sum_{q=0}^{Q}a_q\mu^{q}\,;
\end{equation}
where $h(\mu)$ in (\ref{eq:spectral_filter}) is the frequency response of the filter.

Let $g_k$ denote the desired gain at the distinct sheaf frequency $\mu_k$. 
Since a polynomial filter satisfies
\begin{equation}
h_k
\coloneqq
h(\mu_k)
=
\sum_{q=0}^{Q}a_q\mu_k^q,
\qquad k=1,\ldots,K\,;
\end{equation}
the filter coefficients can be selected by solving
\begin{equation}
\label{eq:filter_design_ls}
\min_{\{a_q\}_{q=0}^{Q}}
\sum_{k=1}^{K}
\left|
g_k
-
\sum_{q=0}^{Q}a_q\mu_k^q
\right|^2.
\end{equation}
Equivalently, defining the Vandermonde matrix $\bm \Phi \in\reall^{K\times(Q+1)}$ with entries
$[\bm \Phi]_{kq}=\mu_k^q$, $q=0,\ldots,Q$, and $\mathbf g=[g_1,\ldots,g_K]^\top$, one obtains
\begin{equation}
\label{eq:filter_design_vandermonde}
\mathbf a^\star
=
\argmin_{\mathbf a\in\reall^{Q+1}}
\|\mathbf g-\bm\Phi\mathbf a\|_2^2=\bm\Phi^\dagger\mathbf g,
\end{equation}
where $(\cdot)^\dagger$ denotes the Moore--Penrose pseudoinverse. 
If $Q\ge K-1$ and the frequencies $\{\mu_k\}$ are distinct, the mask can be interpolated exactly. 
Importantly, the filter design depends only on the distinct frequencies ${\mu_k}$ and the associated spectral projectors ${\mathbf P_k}$, without requiring a choice of eigenbasis within degenerate eigenspaces. 
In particular,
\begin{equation}
h(\mathbf L_{\nsrep})=\sum_{k=1}^{K} h_k\mathbf P_k,
\end{equation}
meaning that the same gain $h_k$ is applied to the entire intrinsic spectral component $\mathbf P_k\vecs$. 
Thus, filtering is intrinsically defined at the eigenspace level and is invariant to the particular orthonormal basis chosen within each eigenspace.

The intertwining relation established in \cref{th:intertwining} naturally extends to every polynomial filter.
\begin{corollary}[Spectral correspondence of polynomial filters]
\label{cor:filter_intertwining}

Let
\begin{equation}
    h(\lambda)=\sum_{q=0}^{Q}a_q\lambda^q.
\end{equation}
Then
\begin{equation}
\label{eq:filter_intertwining}
h(\mathbf{L}_{\ns})
\,\mathbb{D}_0
=
\mathbb{D}_0\,
h(\mathbf{L}_{\nsrep}).
\end{equation}
Consequently, filtering and lifting commute: 
filtering the representation signal and subsequently lifting it to the signal sheaf produces exactly the same result as first lifting the signal and then filtering it on the signal sheaf.
\end{corollary}
\begin{proof}
See Supp. S2.
\end{proof}
The previous corollary highlights one of the main practical advantages of the proposed framework. 
Since every function of the sheaf Laplacian commutes with the natural transformation, spectral processing can be performed directly on the lower-dimensional representation sheaf without loss of information. 
Consequently, the computational complexity is governed by the dimensions of the representation stalks rather than those of the original signal stalks. 
More generally, the representation sheaf provides the natural computational domain for spectral processing over network sheaves, combining reduced dimensionality with complete preservation of the spectral structure.

%% file: sections/5-sheaf-sampling.tex

\section{Sampling over a network sheaf}\label{sec:sampling}

Sampling concerns the recovery of a signal from a subset of its observed values. 
In GSP, perfect recovery from samples collected at a subset of the nodes requires the signal to be bandlimited with respect to the graph Laplacian \cite{tsitsvero2016signals}. 
We extend this framework to network sheaves. 
Unlike GSP, where each sample is a scalar associated to a node, in our case a sheaf associates a vector $\x_i\in\ns(i)$ with each node. 
Consequently, sampling selects a subset of the overall $D$-dimensional $0$-cochain $\x$ and might involve the selection (or not) of some nodes and the observation of a subset of entries in each local vector $\x_i\in\ns(i)$. 

Given a sampling set $\mathcal S\subseteq\{1,\ldots,D\}$, we define the binary sampling matrix
\begin{equation}
    \boldsymbol{\Psi}_{\mathcal S}\in \{0,1\}^{|\mathcal S|\times D},    
\end{equation}
whose rows are the canonical vectors corresponding to the sampled coordinates. 
The observed signal is therefore
\begin{equation}
\label{eq:sampling_measurements}
\x_{\mathcal S}
=
\boldsymbol{\Psi}_{\mathcal S}\x\, .
\end{equation}
Equivalently, sampling can be represented through the orthogonal projector
\begin{equation}
\label{eq:sampling_projector}
\mathbf M_{\mathcal S}
=
\boldsymbol{\Psi}_{\mathcal S}^{\top}
\boldsymbol{\Psi}_{\mathcal S},
\end{equation}
while $\mathbf M_{\mathcal S}^{\perp}=\eye{D}-\mathbf M_{\mathcal S}$ projects onto the unobserved coordinates.

\subsection{Sampling and recovery of sheaf signals}
Let $\{\mu_k\}_{k=1}^{K}$ denote the distinct eigenvalues of the signal sheaf Laplacian $\mathbf L_{\ns}$, with associated eigenspaces $\{\mathcal F_k\}_{k=1}^{K}$ and $\mathbb G_0$-orthogonal projectors $\{\mathbf P_k\}_{k=1}^{K}$. 
For a frequency index set $\mathcal K\subseteq\{1,\ldots,K\}$, define the bandlimiting projector
\begin{equation}
\label{eq:bandlimiting_projector}
\mathbf B_{\mathcal K} =\sum_{k\in\mathcal K}\mathbf P_k\,.
\end{equation}
The corresponding bandlimited subspace is
\begin{equation}
\label{eq:bandlimited_subspace}
\mathcal B_{\mathcal K}
=
\operatorname{Im}(\mathbf B_{\mathcal K})\,,
\end{equation}
whose dimension is
$b_{\mathcal K}=\operatorname{rank}(\mathbf B_{\mathcal K})=\sum_{k\in\mathcal K}m_k,$
where
$m_k=\dim(\mathcal F_k)$.
A signal $\x\in C^0(\graph;\ns)$ is said to be bandlimited to the frequency set $\mathcal K$ if
\begin{equation}
\label{eq:bandlimited_signal}
\mathbf B_{\mathcal K}\x=\x\,.
\end{equation}
For computational purposes, let $\mathbf V_{\mathcal K}\in\mathbb R^{D\times b_{\mathcal K}}$ contain any $\mathbb G_0$-orthonormal basis of $\mathcal B_{\mathcal K}$.
Then
\begin{equation}
    \mathbf B_{\mathcal K}
    =
    \mathbf V_{\mathcal K}
    \mathbf V_{\mathcal K}^{\top}
    \mathbb G_0\,,    
\end{equation}
and every bandlimited signal admits the representation
\begin{equation}
\label{eq:bandlimited_expansion}
\x
=
\mathbf V_{\mathcal K}\boldsymbol{\alpha},
\end{equation}
for a unique coefficient vector $\boldsymbol{\alpha}\in\mathbb R^{b_{\mathcal K}}$.
Notice that the notion of bandlimitedness depends only on the intrinsic spectral subspace $\mathcal B_{\mathcal K}$, not on the particular orthonormal basis $\mathbf V_{\mathcal K}$ used to represent it.

\begin{theorem}[Sampling theorem for sheaf signals]
\label{thm:sheaf_sampling}
Let $\x\in\mathcal B_{\mathcal K}$ and
let $\mathbf V_{\mathcal K}\in\mathbb R^{D\times b_{\mathcal K}}$ be any $\mathbb G_0$-orthonormal basis of $\mathcal B_{\mathcal K}$, where $b_{\mathcal K}=\dim(\mathcal B_{\mathcal K})$. 
Then, $\x$ can be uniquely recovered from the samples
\begin{equation}
\x_{\mathcal S} = \boldsymbol{\Psi}_{\mathcal S}\x
\end{equation}
if and only if
\begin{equation}
\label{eq:sampling_rank}
\operatorname{rank}
\!\left(
\boldsymbol{\Psi}_{\mathcal S}
\mathbf V_{\mathcal K}
\right)
=
b_{\mathcal K},
\end{equation}
or, equivalently, 
$\label{eq:sampling_nullspace}
\mathcal B_{\mathcal K}
\cap
\ker(\boldsymbol{\Psi}_{\mathcal S})
=
\{\mathbf0\}.$
When these conditions hold, $\x$ can be uniquely recovered as
\begin{equation}
\label{eq:sampling_reconstruction}
\x
=
\mathbf V_{\mathcal K}
\left(
\boldsymbol{\Psi}_{\mathcal S}
\mathbf V_{\mathcal K}
\right)^{\dagger}
\x_{\mathcal S}.
\end{equation}
\end{theorem}
\begin{proof}
    See Supp. S2.
\end{proof}

\begin{remark}[Geometric interpretation]
\label{rem:geometric_sampling}
The rank condition in \cref{eq:sampling_rank} holds for arbitrary metric tensors.
When the sampling projector
$\mathbf M_{\mathcal S}
=
\boldsymbol{\Psi}_{\mathcal S}^{\top}
\boldsymbol{\Psi}_{\mathcal S}$
is orthogonal with respect to the Hilbert-space inner product induced by
$\mathbb G_0$ (e.g., when $\mathbb G_0=\eye{D}$, or more generally whenever $\mathbf M_{\mathcal S}\mathbb G_0 = \mathbb G_0\mathbf M_{\mathcal S}$), the recovery condition admits the equivalent characterization
\begin{equation}
\left\|
\mathbf B_{\mathcal K}
\mathbf M_{\mathcal S}^{\perp}
\right\|_{\mathbb G_0}
<1.
\end{equation}
The operator $\mathbf B_{\mathcal K}\mathbf M_{\mathcal S}^{\perp}$ measures the largest bandlimited component that may remain hidden in the unsampled coordinates. 
Thus, the above condition guarantees that no nonzero bandlimited signal is invisible to the sampling operator. 
The proof follows \cite{tsitsvero2016signals}, replacing the Euclidean inner product with the $\mathbb G_0$-induced Hilbert product.
\end{remark}

\begin{corollary}[Minimum number of samples]
\label{cor:minimum_samples}
A necessary condition for the perfect recovery of a sheaf signal
$\x\in\mathcal B_{\mathcal K}$ is
\begin{equation}
\label{eq:minimum_samples}
|\mathcal S|
\geq
b_{\mathcal K},
\end{equation}
where
$b_{\mathcal K}
=
\dim(\mathcal B_{\mathcal K})$
is the bandwidth of the signal.
\end{corollary}
\begin{proof}
See Supp. S2.
\end{proof}

It is worth to point out that the condition in \cref{eq:minimum_samples} is necessary but not sufficient:
the sampled coordinates must also be suitably located so that $\boldsymbol{\Psi}_{\mathcal S}\mathbf V_{\mathcal K}$ has full column rank. 
This distinction is especially important for sheaf signals: 
two sampling sets with the same cardinality may have very different recovery properties because they may select different nodes and different entries within the corresponding stalks.

\subsection{Band splitting induced by the natural transformation $\delta$}
When the signal sheaf is induced by a representation sheaf through a natural transformation, the local signal model induces a natural decomposition of every spectral band. 
This decomposition reveals that only the model-aligned component needs to be recovered whenever the signal is known to belong to the representation subspace. 
Let $\boldsymbol\Pi_\delta:=\mathbb D_0\mathbb H_0^{-1}\mathbb D_0^\top\mathbb G_0$ denote the $\mathbb G_0$-orthogonal projector onto $\im{\mathbb D_0}$.
The projector is well defined whenever $\ns$ arises, via $\delta:\nscoeff\Rightarrow\ns$, as the lift of the representation sheaf $\nscoeff$ (\cref{th:intertwining}).
As we show next, $\boldsymbol\Pi_\delta$ then reveals a natural, model-aligned splitting of every band.

\begin{corollary}[$\delta$-informed splitting of the band]
\label{cor:parsimonious_lift}
Every band $\mathcal B_\mathcal K$ splits $\mathbb G_0$-orthogonally as
\begin{equation}\label{eq:band-splitting}
\begin{aligned}
&\mathcal{B}_{\mathcal{K}}=\im{\mathbf B_\mathcal K^\parallel} \oplus \im{\mathbf B_\mathcal K^\perp}\,,\\
\text{where }&
\mathbf B_\mathcal K^\parallel:=\boldsymbol\Pi_\delta\mathbf B_\mathcal K\,, \quad \mathbf B_\mathcal K^\perp:=(\eye{}-\boldsymbol\Pi_\delta)\mathbf B_\mathcal K\,;    
\end{aligned}
\end{equation}
into a component aligned with the local signal model $\x_i=\D_i\vecs_i$, and one orthogonal to it.
The aligned component admits the explicit form
\begin{equation}
\mathrm{Im}(\mathbf B_\mathcal K^\parallel)=\mathrm{Im}(\mathbb D_0\mathbf U_\mathcal K)\,,
\end{equation}
where $\mathbf{U}_\mathcal K$ collects $\mathbb H_0$-orthonormal bases $\mathbf{U}_k$ of the eigenspaces $\mathcal E_k$ of $\mathbf L_\nscoeff$, $k\in\mathcal K$, with $m'_k\coloneqq\dim(\mathcal E_k)\leq m_k$.
\end{corollary}
\begin{proof}
    See Supp. S2.
\end{proof}

Interestingly, \cref{cor:parsimonious_lift} shows that, whenever the signal of interest is known to obey the local model, i.e., to lie in $\im{\mathbb D_0}$, exact recovery does not require resolving the whole band $\mathcal B_\mathcal K$, but only its aligned component $\mathrm{Im}(\mathbf B_\mathcal K^\parallel)$: 
since $m'_k\le m_k$ for every $k\in\mathcal K$ (often strictly, and possibly $m'_k=0$) the aligned part can have dimension considerably smaller than the nominal bandwidth $b_\mathcal K$.
In practice, \cref{thm:sheaf_sampling} can then be applied with $\mathbf U=\mathbb D_0\mathbf U_\mathcal K$ in place of $\mathbf V_\mathcal K$, targeting $\mathrm{Im}(\mathbf B_\mathcal K^\parallel)$: 
the effective bandwidth to be resolved shrinks from $b_\mathcal K$ to $\sum_{k\in\mathcal K}m'_k$, and so does the number of samples in \cref{cor:minimum_samples} sufficient for perfect recovery.
\cref{subsec:app_sampling} illustrates this reduction empirically.

Before concluding this section, it is worth to point out that, unlike GSP, where recoverability depends only on the graph topology, in network sheaves it is jointly determined by the graph topology, the restriction maps, and the metric tensors defining the sheaf geometry. 
\cref{thm:sheaf_sampling} therefore establishes a direct connection between the algebraic structure of the sheaf and the recoverability of bandlimited signals.

\subsection{Sampling Strategy}
\label{sec:sampling_strategies}
The recovery condition of \cref{cor:minimum_samples} characterizes when a given $\mathcal{S}$ enables perfect recovery, but does not prescribe how to choose it.
Let $\mathbf V\in\reall^{D\times b}$ denote any $\mathbb G_0$-orthonormal basis of a target subspace $\mathcal B\subseteq C^0(\graph;\ns)$, $b=\dim\mathcal B$. 
By default $\mathbf V=\mathbf V_\mathcal K$ and $b=b_\mathcal K$, the full band of \cref{thm:sheaf_sampling}, but $\mathbf V$ may equally be taken as the $\delta$-aligned basis $\mathbb D_0\mathbf U_\mathcal K$ of \cref{cor:parsimonious_lift}, with $b=b_\mathcal K^\parallel:=\sum_{k\in\mathcal K}m_k'=\dim{\im{\mathbf B_\mathcal K^\parallel}}$, whenever the signal is known to lie in $\im{\mathbb{D}_0}$.
The design problem is to find a minimum-size sampling set $\mathcal S$ such that $\mathrm{rank}(\boldsymbol{\Psi}_\mathcal{S}\mathbf V) = b$, which is NP-hard in general~\cite{dilorenzo2018sampling}.

\spara{Greedy algorithm.}
Given as minimum sampling budget the signal bandwidth $b$ (cf. \cref{cor:minimum_samples}), we seek the sampling set that maximizes the rank of the sampling operator restricted to the bandlimited subspace (cf. \cref{eq:sampling_rank}):
\begin{equation}
\begin{aligned}
\max_{\mathcal S\subseteq\mathcal I}\quad
& \mathrm{rank}(\boldsymbol{\Psi}_{\mathcal S}\mathbf V)\\
\text{s.t.}\quad
& |\mathcal S|=b\,;
\end{aligned}
\label{eq:sampling_optimization}
\end{equation}
where $\mathcal I$ denotes the set of all sheaf signal entries. Defining
$f(\mathcal S)=\mathrm{rank}(\boldsymbol{\Psi}_{\mathcal S}\mathbf V)$,
the objective is the rank function of a linear matroid, and is therefore monotone and submodular~\cite{nemhauser1978analysis}. Consequently, a greedy algorithmic approach provides a $(1-1/e)$-approximation to the optimal solution.

\spara{Tie-breaking.}
Identify each global coordinate in $\{1,\ldots,D\}$ with a pair $(i,k)$, $i\in\vertexset$, $k=1,\ldots,d_i$, via the stalk offsets, and let $\widetilde{\mathbf V}:=\mathbb G_0^{1/2}\mathbf V$, which is Euclidean-orthonormal since $\widetilde{\mathbf V}^\top\widetilde{\mathbf V}=\mathbf V^\top\mathbb G_0\mathbf V=\eye{b}$.
Among all pairs that increase the rank, ties are broken by the row norm
\begin{equation}
  \eta_{i,k} = \bigl\|\mathbf{e}_k^\top\widetilde{\mathbf V}[\mathrm{stalk}_i,:]\bigr\|_2,
  \label{eq:tiebreak}
\end{equation}
which measures how much of the target subspace's energy raw coordinate $k$ at node $i$ captures.
Crucially, the atoms of selection are always \emph{raw stalk coordinates}, never dictionary atoms: $\boldsymbol{\Psi}_\mathcal{S}$ remains the plain binary sampling matrix of \cref{eq:sampling_measurements,eq:sampling_projector} throughout, regardless of which target basis $\mathbf V$ is used.
\Cref{alg:greedy_sampling} summarizes the greedy sampling strategy.

\begin{algorithm}[t]
  \caption{Greedy Sampling for Bandlimited Sheaf Signals}
  \label{alg:greedy_sampling}
  \begin{algorithmic}
    \REQUIRE $\mathbb G_0$-orthonormal basis $\mathbf V\in\reall^{D\times b}$ of the target subspace $\mathcal B$
             (e.g.\ $\mathbf V_\mathcal K$, \cref{thm:sheaf_sampling}, or $\mathbb D_0\mathbf U_\mathcal K$, \cref{cor:parsimonious_lift})
    \ENSURE  Sampling set $\mathcal{S}$, operator $\boldsymbol{\Psi}_\mathcal{S}$
    \STATE $\mathcal{S} \leftarrow \emptyset$,\quad $r \leftarrow 0$,\quad
           $\mathcal{C} \leftarrow \{(i,k) : i \in \vertexset,\,k = 1,\ldots,d_i\}$,\quad
           $\widetilde{\mathbf V}\leftarrow\mathbb G_0^{1/2}\mathbf V$
    \STATE $\eta_{i,k} \leftarrow \|\mathbf{e}_k^\top\widetilde{\mathbf V}[\mathrm{stalk}_i,:]\|_2$
          \textbf{for all} $(i,k) \in \mathcal{C}$
    \WHILE{$r < b$}
      \STATE $\mathcal{C}^+ \leftarrow \{(i,k) \in \mathcal{C} \setminus \mathcal{S} :
             \mathrm{rank}(\boldsymbol{\Psi}_{\mathcal{S}\cup\{(i,k)\}}
             \mathbf V) > r\}$
      \IF{$\mathcal{C}^+ = \emptyset$}
        \STATE \textbf{break}
      \ENDIF
      \STATE $(i^*,k^*) \leftarrow \arg\max_{(i,k)\in\mathcal{C}^+} \eta_{i,k}$
      \STATE $\mathcal{S} \leftarrow \mathcal{S} \cup \{(i^*,k^*)\}$,\quad $r \leftarrow r+1$
    \ENDWHILE
    \RETURN $\mathcal{S}$,\quad $\boldsymbol{\Psi}_\mathcal{S}\in\{0,1\}^{|\mathcal{S}|\times D}$ (\cref{eq:sampling_measurements})
  \end{algorithmic}
\end{algorithm}

%% file: sections/6-applications.tex

\section{Empirical assessment}\label{sec:empirical-assessment}

Our framework is validated through three numerical examples. 
The first investigates filtering on synthetic data, the second demonstrates the proposed framework on a real-world multi-view motion capture dataset, and the third shows an application of sampling for financial portfolio reconstruction.

\subsection{Filtering signals from correlated latent factors}\label{subsec:app_filtering_synth}
We consider the network sheaf topology shown in \cref{fig:sheaf_topo_synth}: 
two cliques of $5$ nodes each, connected by $3$ bridge edges.
Each clique $k\in\{1,2\}$ is associated with a latent factor $\y_k = C_k\m + \boldsymbol{\varepsilon}_k$ in $\reall^r$ ($r=10$), where $\m \sim \mathcal{N}(\zeros, \eye{r})$ is a confounder shared across both cliques, inducing correlation between $\y_1$ and $\y_2$, and $\boldsymbol{\varepsilon}_k\sim \mathcal{N}(\zeros, \eye{r})$ is an idiosyncratic component specific to clique $k$, independent of $\m$ and of $\boldsymbol{\varepsilon}_{k'}$ for $k'\neq k$.
The scalars $C_1, C_2 \in[0,1]$ control the correlation between the two factors. 
For simplicity, we set $C_1=C_2=C$.
Each node $i$ carries a raw signal $\x_i = \myH_i\y_{k(i)}$, where $k(i)\in\{1,2\}$ denotes the clique containing $i$ (i.e., $k(i)=1$ for $i\in\{0,\dots,4\}$ and $k(i)=2$ for $i\in\{5,\dots,9\}$), and $\myH_i\in\reall^{d\times r}$ ($d=20$) is a random linear embedding, drawn independently across nodes with no correlation between them.

On intra-clique edges, the consistency condition requires both endpoints to recover the same latent factor $\y_k$, giving $\restrictionmap{i}{e} = \mathrm{pinv}(\myH_i) \in \reall^{r\times d}$. 
On bridge edges, the consistency condition instead requires the two endpoints to agree on the MMSE estimate of the shared component $C\m$, yielding $\restrictionmap{i}{e} = \beta\,\mathrm{pinv}(\myH_i)$, with $\beta = \frac{C^2}{C^2+1} = \mathrm{Cov}(C\m,\y_i)/\mathrm{Var}(\y_i)$ the optimal linear regression coefficient for estimating $C\m$ from $\y_i$.
The representation sheaf $\nscoeff$ is then obtained via $\restrictionmapscoeff{i}{e} = \restrictionmap{i}{e}\D_i$, where $\D_i\in\stiefel{d}{c}$ ($c=10$) is a local DCT dictionary, ensuring naturality by construction. 
From the restriction maps $\restrictionmapscoeff{i}{e}$, we assemble the sheaf Laplacian $\mathbf{L}_\nscoeff$ and compute its eigendecomposition.
Let $\{\mu_k\}$ and $\{\mathcal{E}_k\}$ denote the distinct eigenvalues and eigenspaces of $\mathbf{L}_\nscoeff$.
A bandlimited test signal is synthetized from $K_\mathrm{sig}=15$ selected eigenspaces using a decaying energy profile $\{w_k\}$:
\begin{equation}\label{eq:ex_signal_generation}
    \vecs = \sum_{k \in \mathcal{K}} w_k \, \U_k \boldsymbol{\alpha}_k\,,
    \qquad \boldsymbol{\alpha}_k \sim \mathrm{Uniform}\bigl(\mathcal{S}^{|\mathcal{E}_k|-1}\bigr)\,,
\end{equation}
where $\mathcal{K}$ is the index set of the $K_\mathrm{sig}$ selected eigenspaces and $\mathcal{S}^{|\mathcal{E}_k|-1}$ denotes the unit sphere in $\reall^{|\mathcal{E}_k|}$, ensuring basis-invariance within each eigenspace.
By construction, $\vecs$ is bandlimited with respect to $\mathbf{L}_\nscoeff$---its energy is supported on $\mathcal{K}$---but is not in general a global section of $\nscoeff$, as $\mathcal{K}$ need not be restricted to the kernel eigenspaces $\{\mu_k = 0\}$.
The signal is then lifted to the raw domain stalk-wise via $\x_i = \D_i\vecs_i$.

\begin{figure}[t]
    \centering
    \includegraphics[width=0.9\linewidth]{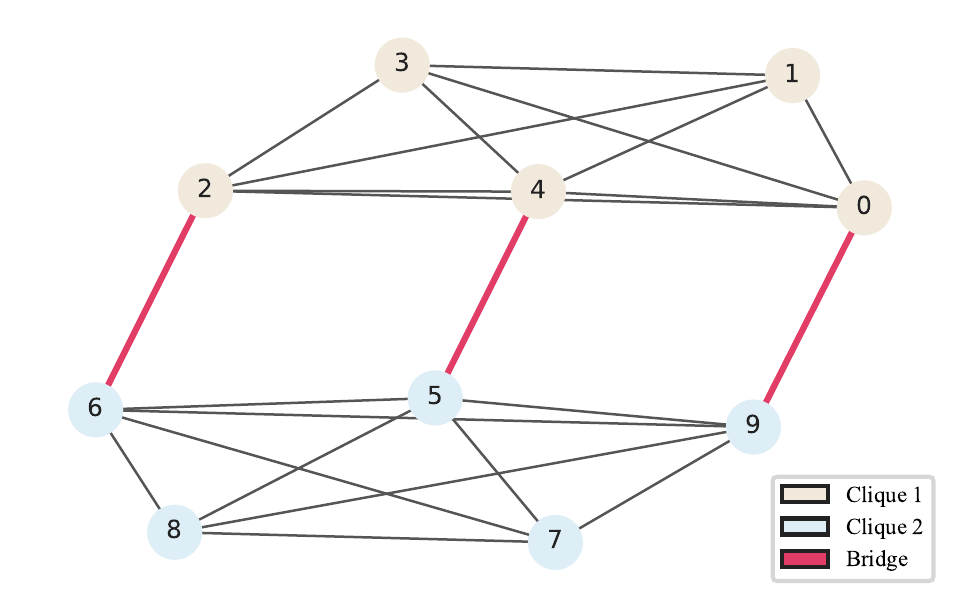}
    \caption{Sensor network topology of the filtering experiment.}
    \label{fig:sheaf_topo_synth}
\end{figure}

Once we generate a set of bandlimited signals, we corrupt them with addiditve white Gaussian noise at different signal-to-noise ratio (SNR) levels. 
We aim at assessing the filtering capabilities of the two sheaf representations $\mathbf{L}_J$ and $\mathbf{L}_F$.
Thus, we implement spectral filtering via hard-thresholding: 
we rank the eigenspaces by the signal's energy in each, then project the signal onto the highest-ranked ones. 
The number of eigenspaces is chosen by exploring the tradeoff between bandwidth and reconstruction error, and reporting the optimal value.
The filtering is implemented in the representation domain, so that the reconstruction error is computed after lifting to the signal domain. 
We compare the performance of filtering with respect to $\mathbf{L}_J$ with \emph{(i)} filtering in the raw domain with $\mathbf{L}_F$, \emph{(ii)} filters based only on local dictionaries projection, and \emph{(iii)} a GSP baseline using a graph Laplacian learned from additional noiseless training signals with the algorithm from~\cite{dong2016learning}. 
For the latter, we consider two variants to ensure a fair comparison in terms of filtering-complexity:
\emph{(i)} \emph{GLSigRep-InterOnly}, where intra-stalk edges are forbidden,
and \emph{(ii)} \emph{GLSigRep-TopoMasked}, where inter-stalk edges are also constrained to adhere to the topology in \cref{fig:sheaf_topo_synth}.

\begin{figure}[t]
    \centering
    \includegraphics[width=1\linewidth]{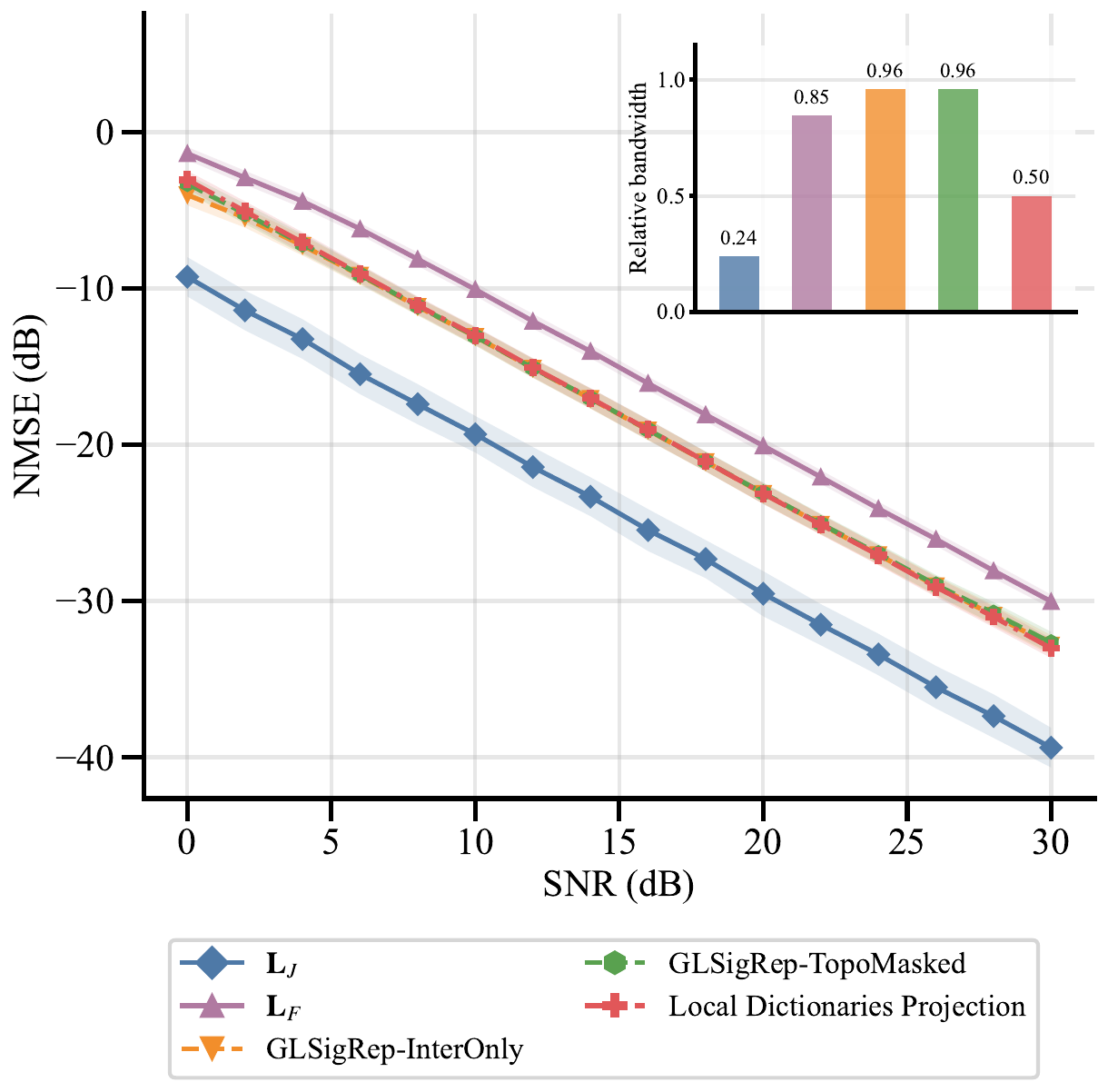}
    \caption{Reconstruction error (left panel) and fraction of modes at the optimal bandwidth (top-right corner).}
    \label{fig:synth_filtering}
\end{figure}

\Cref{fig:synth_filtering} reports the NMSE against the SNR for all considered approaches, with the top-right inset showing each method's optimal relative bandwidth (i.e., the fraction of eigenvectors retained). 
The representation sheaf $\mathbf{L}_J$ strongly outperforms every baseline, both in NMSE and filter sparsity. 
The GSP baselines fail to localize the signal spectrally:
they require a much larger bandwidth, yet their NMSE does not improve on simply projecting onto the local dictionaries. 
The signal sheaf $\mathbf{L}_F$ likewise fails to capture the signal's bandlimitedness, since the spectral intertwining of \cref{th:intertwining} does not hold here: 
the restriction maps of \ns are not derived from \nscoeff, but the converse. 
Filtering with the two sheaves thus yields different results, with \ns performing significantly worse.

\subsection{Denoising inter-frame displacement in CMU Panoptic}
\label{sec:panoptic}
We consider the CMU Panoptic Studio dataset\footnote{\url{http://domedb.perception.cs.cmu.edu}}, sequence \texttt{171204\_pose1}, which contains synchronized recordings of a human subject observed by a calibrated multi-camera system. 
For each frame, the dataset provides the $3$D coordinates of the $19$ body joints (COCO19 format, world coordinates in cm) and the corresponding synchronized observations from $31$ calibrated HD cameras. 
Each camera $c$ is described by its intrinsic calibration 
$\mathbf{K}_c$ and extrinsic parameters $(\mathbf{R}_c,\mathbf{t}_c)$, defining the nonlinear projection $\pi_c(\mathbf{K}_c,\mathbf{R}_c,\mathbf{t}_c):\mathbb{R}^3\rightarrow\mathbb{R}^2$. 
The signal of interest is the inter-frame displacement of the body joints between two consecutive frames $t$ and $t+\Delta$.

The graph is instantiated at each reference frame $t$ and comprises two distinct sets of nodes. The first consists of one node for each skeleton joint $j$, with stalk $\ns(j)\cong\mathbb{R}^{3}$ carrying the corresponding $3$D inter-frame displacement $\mathbf{d}_j(t,\Delta)=\mathbf{x}_j(t+\Delta)-\mathbf{x}_j(t)$. 
The second consists of one node for each visible joint-camera pair $(c,j)$, with stalk $\ns(c,j)\cong\mathbb{R}^{2}$ carrying the corresponding image displacement $\mathbf{u}_{cj}=\pi_c(\mathbf{x}_j(t+\Delta))-\pi_c(\mathbf{x}_j(t))$. 
Three families of edges are considered: 
\emph{bone edges}, connecting adjacent joints in the kinematic skeleton; 
\emph{projection edges}, linking each joint $j$ to its image observations $(c,j)$; 
and \emph{view edges}, connecting observations $(c_1,j)$ and $(c_2,j)$ of the same joint across different cameras. 
Bone edges $e_b=(j_1,j_2)$ are associated with the edge stalk $\ns(e_b)\cong\mathbb{R}^{3}$ and restriction maps $\mathbf{F}_{j_1\trianglelefteq e_b}=\mathbf{F}_{j_2\trianglelefteq e_b}=\mathbf{I}_3$, so that consistency enforces locally translational motion between adjacent joints in the canonical coordinate system. 
Projection edges $e_p=(j,(c,j))$ are instead frame-dependent. 
They are associated with the edge stalk $\ns(e_p)\cong\mathbb{R}^{2}$, the restriction map $\mathbf{F}_{j\trianglelefteq e_p}=\T_{cj}=\left.\frac{\partial \pi_c}{\partial \mathbf{x}}\right|_{\mathbf{x}_j(t)}\in\mathbb{R}^{2\times 3}$ obtained by linearizing the camera projection around the reference joint position $\mathbf{x}_j(t)$, and $\mathbf{F}_{(c,j)\trianglelefteq e_p}=\mathbf{I}_2$. 
Since $\ker\left(\mathbf{F}_{j\trianglelefteq e_p}\right)$ coincides with the viewing ray $\mathbf{r}_{cj}$ of camera $c$ passing through $\mathbf{x}_j(t)$, a single image observation determines the $3$D displacement only up to its component along the viewing direction. 
To encode the complementary information provided by multiple cameras, view edges connect observations of the same joint across different cameras. 
These edges are associated with a one-dimensional stalk $\ns(e)\cong\mathbb{R}$ and restriction maps
$\mathbf{F}_{(c_1,j)\trianglelefteq e}=\mathbf{m}^{\top}\T_{c_1j}\T_{c_2j}^{\dagger}$ and
$\mathbf{F}_{(c_2,j)\trianglelefteq e}=\mathbf{m}^{\top}$,
where $\mathbf{m}\perp\mathbf{e}_{c_2}=\T_{c_2j}\mathbf{r}_{c_1j}$ and $\|\mathbf{m}\|=1$. 
Up to a scaling factor, this is the unique linear relation between the two image displacements that is satisfied by every underlying $3$D displacement, and it degenerates only when the two camera centers are collinear with the observed joint. 
The reduced edge dimension ($d_e=1<d_i=d_j=2$) also illustrates the interpretation of low-dimensional edge stalks as compression bottlenecks discussed in \cref{sec:primer}. 
Finally, the $3$D and image displacement signals are normalized independently to unit norm before constructing the sheaf signals.

\begin{figure}
    \centering
    \includegraphics[width=1\linewidth]{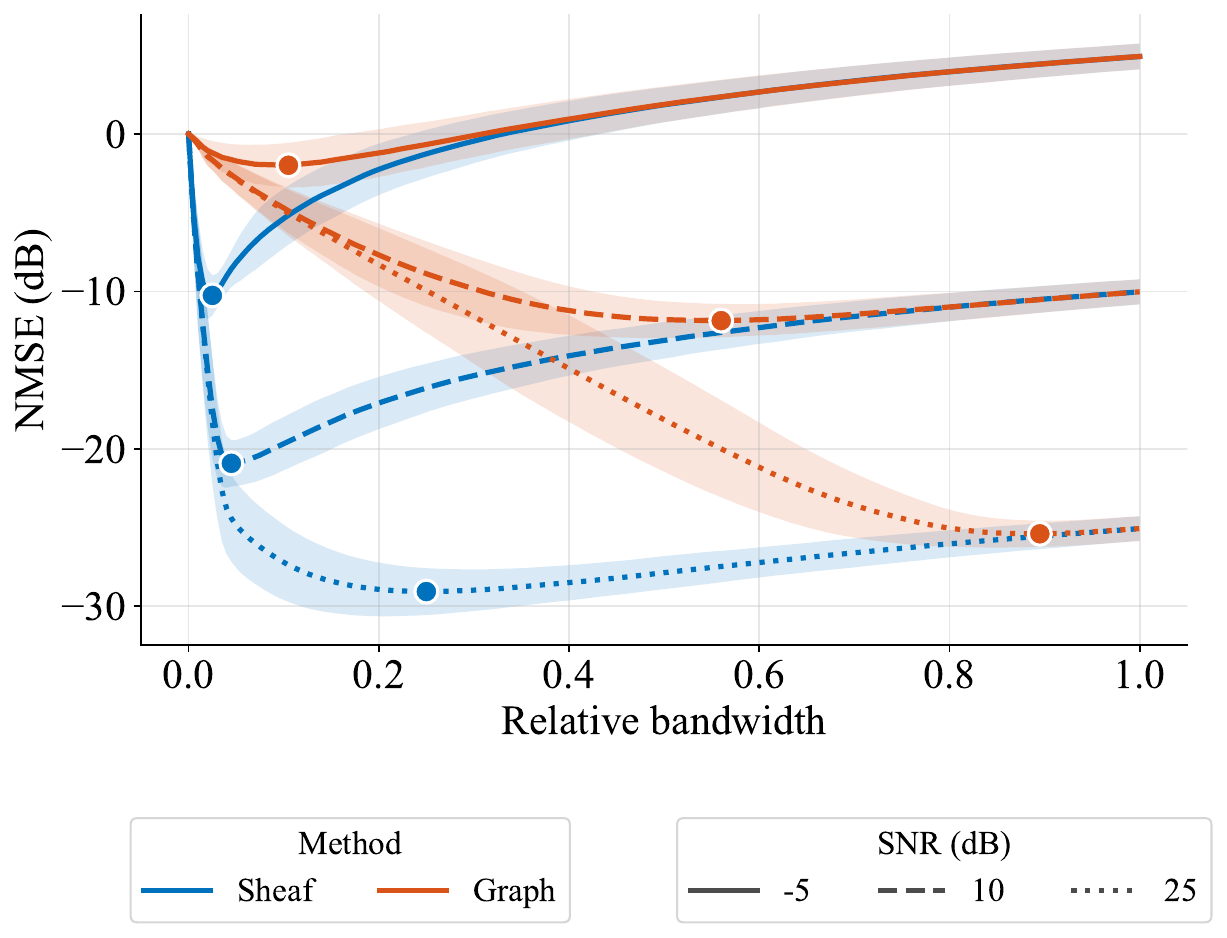}
    \caption{NMSE versus relative bandwidth, for different SNR and methods, over the Panoptic CMU experiment.}
    \label{fig:panoptic}
\end{figure}

We compare the sheaf Laplacian $\mathbf{L}_F$ against an isotropic, GSP-like operator acting on the joint coordinates. 
This baseline is the same network used in the sheaf construction, restricted to the nodes representing the skeletal joints and stripped of the camera-view enrichment layer; the operator can only enforce local translational consistency. 
The comparison isolates denoising performance on the signal of interest: 
the joint displacement between consecutive frames.
\Cref{fig:panoptic} shows the NMSE versus the relative bandwidth for top-$k$ spectral filtering at different SNR levels when using the signal sheaf \ns and the GSP baseline, averaged over $8$ noise realizations and $100$ frames.
Since the sheaf depends on the reference frame, it is reconstructed at each frame by linearizing the camera projections around the current joint positions.
\ns consistently achieves a lower NMSE than the GSP baseline across all bandwidth and SNR values, reflecting a more compact spectral representation of the signal. 
Indeed, it allows exploiting multi-view redundancy via the camera projections encoded in its restriction maps, resulting in more robust denoising.

\subsection{Recovering financial portfolios from sampled observations}\label{subsec:app_sampling}
We consider a set $\mathcal{U}$ of financial stocks\footnote{The constituents of S\&P~100 Index.}, with $|\mathcal{U}|=101$.
Five investors, labeled from $A$ to $E$, buy and sell subsets of the stocks in $\mathcal{U}$. 
Denote by $\mathcal{U}_i \subset \mathcal{U}$ the subset corresponding to the $i$-th investor, consisting of $d_i$ stocks, and let the $k$-th entry of a vector $\x_i \in \reall^{d_i}$ be the fraction of the capital invested by investor $i$ corresponding to the stock $k \in \mathcal{U}_i$. 
Further, the above subsets have heterogeneous size: 
$d_A=24,\, d_B=30,\, d_C=52,\, d_D=71,\, d_E=21$. 
We then define a graph \graph with five nodes $\vertexset=\{A, B, C, D, E\}$, and with line topology $A-B-E-D-C$. 
For each edge $e=(i,j)$ of \graph, $\mathcal{U}_i \cap \mathcal{U}_j \neq \emptyset$. 
In addition, $\mathcal{U}_A$ and $\mathcal{U}_B$ share no stocks with either $\mathcal{U}_C$ or $\mathcal{U}_D$, thus $\mathcal{U}_E$ acts as a bridge between these two clusters.

Now, we build the sheaf \ns by defining the node stalks as $\ns(i)\cong\reall^{d_i}$, where the node signals are those vectors $\x_i$, $i \in \vertexset$.
Thus, the $0$-cochain is $\x \in \reall^{198}$.
Regarding the edge stalks, for each $e=(i,j)$ with $d_i<d_j$, we let $\ns(e)\equiv \ns(j)$.
For instance, for $e=(C,D)$, $\ns(e)\equiv \ns(D)$. The restriction maps compare two neighboring investors over the stocks they hold in common; we make this precise below, once we introduce the data-driven dictionaries, since that is how these maps are actually constructed. 
Each $\x_i$ can be equivalently represented over a dictionary $\D_i \in \stiefel{d_i}{c}$, whose columns represent principal components--in the financial domain interpreted as \emph{statistical risk factors}--explaining most of the variance of the data (here we set $c=5$). 
We build the dictionaries $\D_i$ according to standard statistical factor models~\cite{connor1986performance} using the time series of daily returns of the $101$ stocks over the period from Jan $1^{\text{st}}$, $2022$ to Dec $31^{\text{st}}$, $2024$, gathered from Yahoo Finance.
The local signal models are thus $\x_i = \D_i\vecs_i$, where $\vecs_i$ is the investor-specific representation of the fractions of invested money in $\x_i$ over the $5$ statistical risk factors.

According to the local signal models and to the constructed signal sheaf \ns, the node stalks of the representation sheaf \nscoeff are $\nscoeff(i)\cong\reall^c$, $i \in \vertexset$, with valuations corresponding to $\vecs_i$. 
Thus, the $0$-cochain is $\vecs \in \reall^{25}$.
Similarly to \ns, for each edge, the edge stalk of \nscoeff coincides with the representation stalk of the same reference node $j$ already chosen for that edge in \ns, so that $\D_e=\D_j$ and $\restrictionmapscoeff{j}{e}=\eye{c}$.
Since $\D_i$ and $\D_j$ correspond to different stock subsets, we compare two investors' representations through the stocks they hold in common. 
Specifically, we pad each dictionary with zero rows to the full set $\mathcal U$, obtaining $\widetilde{\D}_i$ and $\widetilde{\D}_j$.
Then, we set representation-level restriction map $\restrictionmapscoeff{i}{e}=\widetilde{\D}_j^\top \widetilde{\D}_i$.
The restriction maps of \ns are obtained by lifting the representation-level restriction through the dictionaries via \cref{eq:metric_lift} (metric tensors here are the identity), thus guaranteeing exact spectral correspondence. 
The resulting sheaf Laplacian $\mathbf{L}_\nscoeff$ has kernel with dimension $5$, and $20$ one-dimensional eigenspaces, for $21$ eigenspaces in total. 
By \cref{cor:spectral_correspondence}, every eigenvalue of $\mathbf{L}_\nscoeff$ is also an eigenvalue of $\mathbf{L}_\ns$, and the corresponding $21$ eigenspaces are embedded isometrically into those of $\mathbf{L}_\ns$. 
The kernel of $\mathbf{L}_\ns$, however, has dimension $178$: 
only $5$ of these directions come from the lifted kernel of \nscoeff, the remaining $173$ being orthogonal to every dictionary.

A bandlimited test signal is then generated on $\nscoeff$ following the eigenspace-uniform procedure of \cref{eq:ex_signal_generation}, selecting $K_{\mathrm{sig}}=5$ eigenspaces of $\mathbf{L}_\nscoeff$ (the $5$-dimensional kernel plus four simple nonzero eigenvalues), for a $b_\mathcal{K}=9$-dimensional target band.
By \cref{cor:spectral_correspondence}, lifting the generated representation to the signal sheaf \ns via the dictionaries yields a signal \x with the same energy as \vecs on the same $K_{\mathrm{sig}}$ eigenspaces.
Diagonalizing $\mathbf{L}_\ns$ directly to isolate the band $\mathcal K$ would be misleading: 
its kernel mixes the $173$ dictionary-invisible directions with the $5$ relevant ones. 
Instead, lifting the $9$ eigenvectors of $\mathbf{L}_\nscoeff$ via $\delta$ yields, by \cref{cor:spectral_correspondence}, the relevant subspace of $\ker(\mathbf{L}_\ns)$ for the band $\mathcal{K}$: 
this $9$-dimensional lifted subspace is the target basis $\mathbf V$ used for sampling in \cref{alg:greedy_sampling}. 
Given the target band $\mathcal{K}$ with $b_\mathcal{K}=9$, our goal is to reconstruct \x from few of its \emph{own raw entries}---i.e., individual portfolio-weight measurements $\x_i[k]$, the fraction invested by $i$ in a specific stock $k\in\mathcal U_i$. 
Concretely, $\x_\mathcal{S}=\Psi_\mathcal{S}^\top\x$ according to \cref{eq:sampling_reconstruction}, where the sampling operator $\Psi_\mathcal{S}$ is obtained through \cref{alg:greedy_sampling}, whose atoms of selection are precisely these raw stock-level entries and never dictionary atoms: 
dictionaries only shape the $9$-dimensional target subspace via the lift above, not which entries of \x are queried.

\begin{figure}[t]
    \centering
    \includegraphics[width=\linewidth]{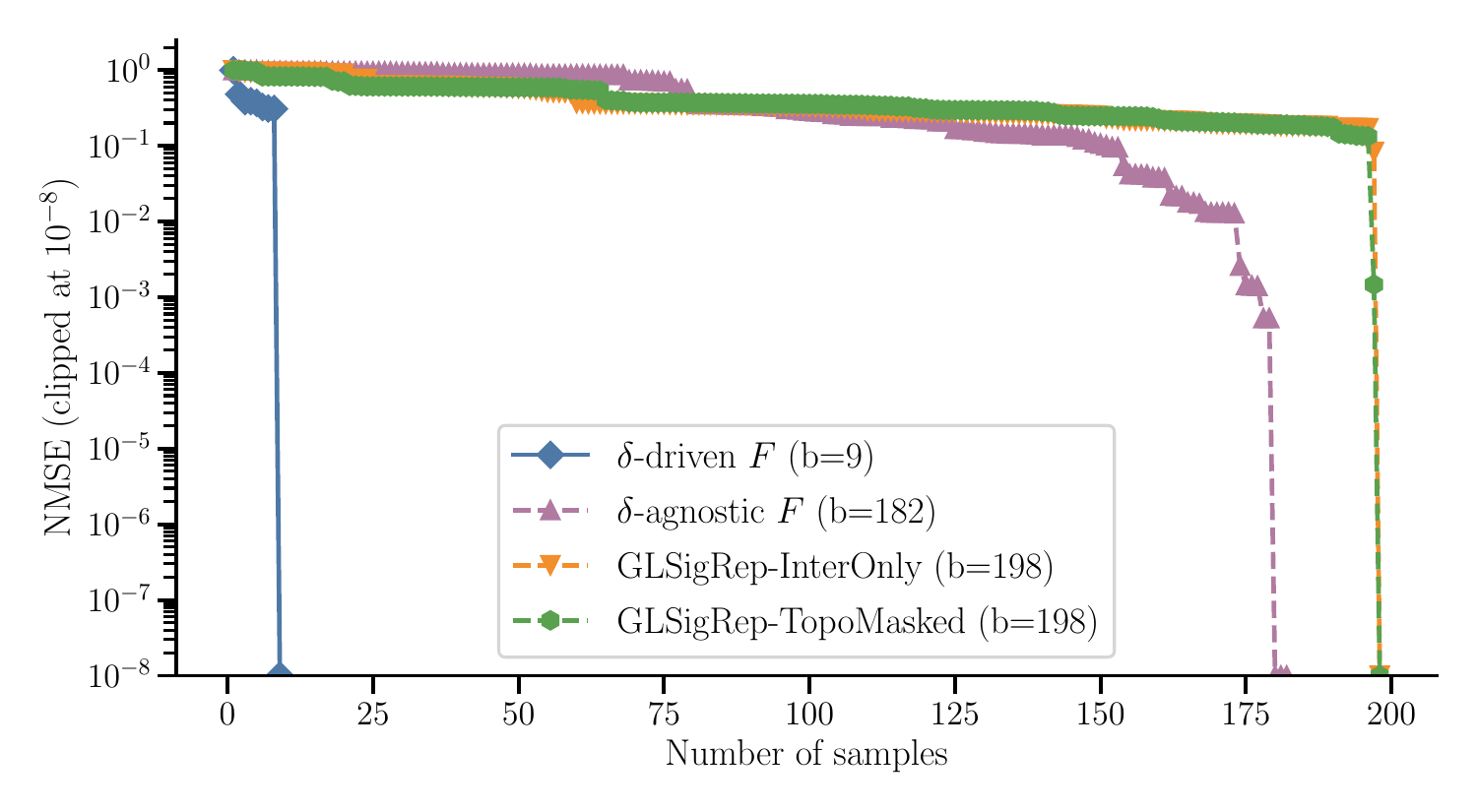}\\
    \footnotesize (a)\\[1ex]
    \includegraphics[width=\linewidth]{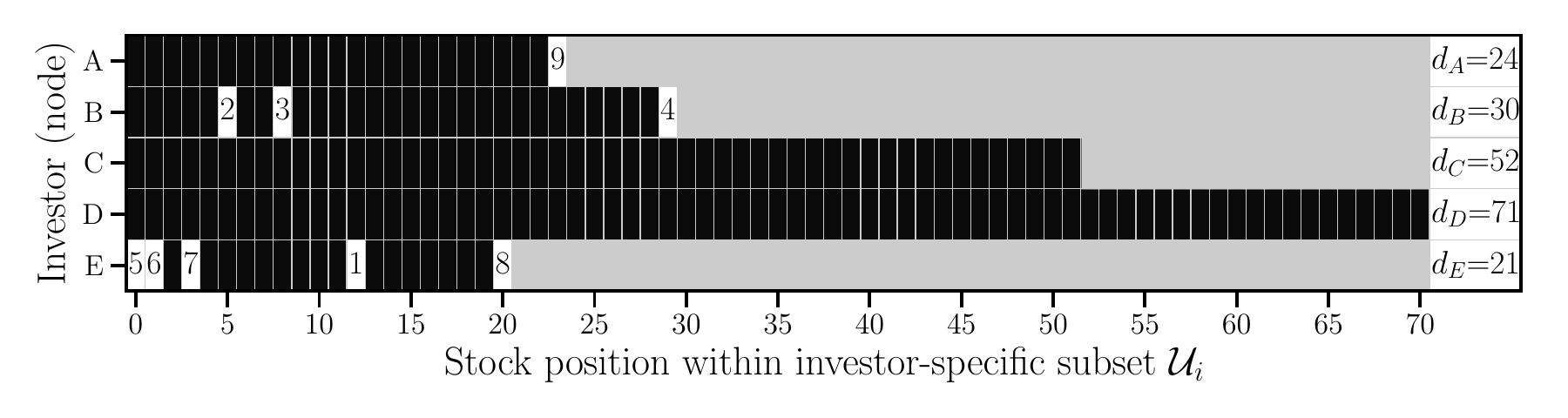}\\
    \footnotesize (b)\\[1ex]
    \includegraphics[width=\linewidth]{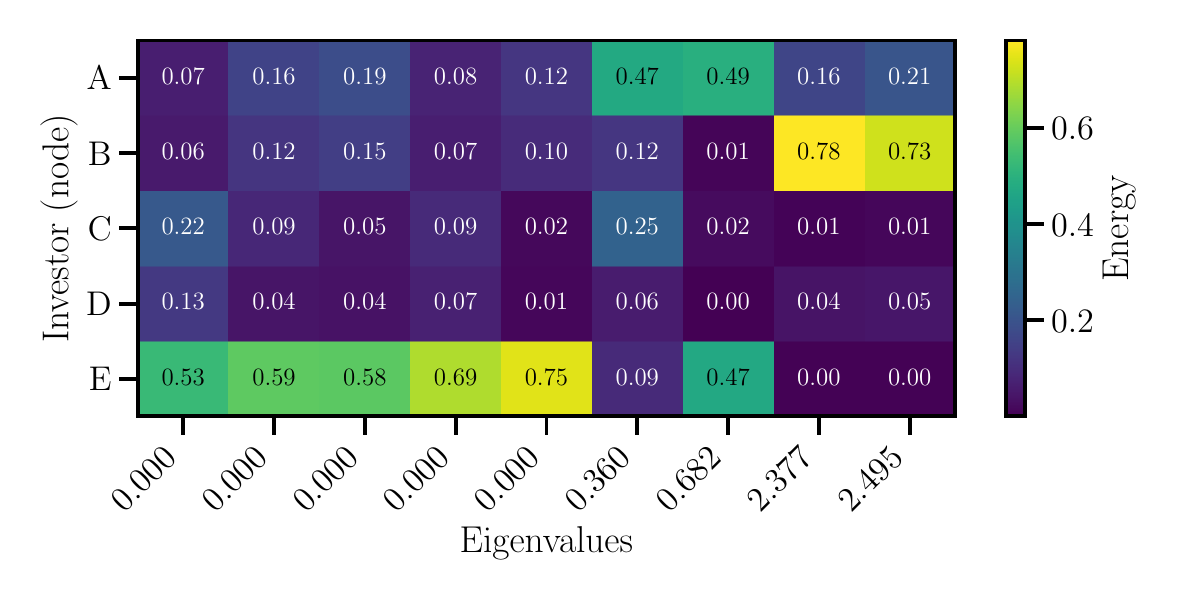}\\
    \footnotesize (c)
    \caption{(a) Reconstruction error (NMSE) vs number of selected samples for the tested methods.
    NMSE is clipped at $10^{-8}$ level for readability.
    Concerning $\delta$-driven sheaf sampling (cf. \cref{cor:parsimonious_lift}), we provide (b) the resulting greedy sample selection (white) with numbers indicating the sampling order, and (c) the energy decomposition across investors.}
    \label{fig:sampling_result}
\end{figure}

The results are collected in \cref{fig:sampling_result}(a)-(c). 
Specifically, \cref{fig:sampling_result}(a) reports the NMSE of signal recovery versus the number of samples, comparing the proposed sheaf approach against the baselines of \cref{subsec:app_filtering_synth}, each fed to the same greedy procedure but with its own target subspace and, consequently, its own minimum sampling budget. 
As we can see from \cref{fig:sampling_result}(a), $\delta$-driven reconstruction (using \cref{cor:parsimonious_lift}) is essentially exact:
the NMSE of the $198$-dimensional sheaf signal \x is on the order $10^{-29}$, consistent with \cref{thm:sheaf_sampling} once $\mathcal S$ achieves $\mathrm{rank}(\Psi_\mathcal{S}^\top\mathbf V)=b_\mathcal{K}=9$, which \cref{alg:greedy_sampling} attains with exactly $9$ raw measurements---the minimum possible by \cref{cor:minimum_samples}. 
Instead, diagonalizing $\mathbf{L}_\ns$ directly and selecting, via \cref{cor:spectral_correspondence}, the eigenspaces overlapping the band $\mathcal{B}_{\mathcal{K}}$ ($\delta$-agnostic \ns) requires $b=182$ raw measurements for exact recovery, $20\times$ more than the $\delta$-driven approach. 
The reason is that the $5$ dictionary-relevant directions from $\ker(\mathbf{L}_\nscoeff)$ are conflated inside the $178$-dimensional kernel of $\mathbf{L}_\ns$, together with the $173$ dictionary-orthogonal ones. 
Any individual eigenvector returned by diagonalizing a degenerate eigenspace is arbitrary within it; recovering the true band exactly therefore requires retaining the eigenspace in full. 
Finally, the GSP baseline, is markedly worse in both the inter-only and the topology-masked variant. 
The NMSE remains order of magnitudes above, and exact reconstruction is only attained once every signal coordinate is sampled.

\cref{fig:sampling_result}(b) illustrates the sampled stock-level positions across investors by the $\delta$-driven sheaf approach. 
As we can see from \cref{fig:sampling_result}(b), the method selects $1$ sample from $A$, $3$ from $B$, and $5$ from $E$. Conversely, $C$ and $D$ are never directly observed. This is consistent with the energy distribution for the generated signal at investor-level, illustrated in \cref{fig:sampling_result}(c), where we can see how signal energy concentrates on $A$, $B$, and $E$, which is precisely why \cref{alg:greedy_sampling} targets them. 
Finally, the signal \x is nearly locally consistent on $E-D-C$, with $\mathrm{TV}_{(C,D)}\!=\!0.089$, $\mathrm{TV}_{(E,D)}\!=\!0.017$.
Thus, the restriction maps alone tightly constrain the investments of $C$ and $D$ from their sampled neighbors, without requiring a direct measurement. 
In contrast, $\mathrm{TV}_{(A,B)}\!=\!0.993$ and $\mathrm{TV}_{(E,B)}\!=\!0.748$, justifying sampling from those investors.

%% file: sections/7-conclusions.tex

\section{Conclusions and future works}\label{sec:conclusions}

We introduced a sheaf-theoretic framework for SP on graphs with heterogeneous local signal spaces. 
By modeling node and edge data as network sheaves valued in \Hilb, the proposed framework extends GSP to settings where local signals differ in dimension, representation, and geometry. 
Building on this model, we developed a unified treatment of spectral representation, filtering, and sampling. 
A key ingredient of the framework is the use of natural transformations to relate signal and representation sheaves, enabling spectral processing in low-dimensional representation spaces while preserving the spectral structure of the original signals. 
These results establish a common foundation for SP over heterogeneous network data.

Several directions deserve further investigations. 
An important theoretical question is the interplay between local and global bandlimitedness, where signals are simultaneously compressible within each stalk and bandlimited over the sheaf. 
Furthermore, in this paper we assumed the restriction maps and the graph topology to be known a priori.
An interesting ongoing activity is how to learn both from data. 
On the learning side, a further interesting direction is task-oriented sheaf learning, where the sheaf structure is jointly optimized with downstream SP objectives, such as denoising, interpolation, or classification. 
Ultimately, extending the framework to cell complexes of arbitrary order would allow capturing multiway relations among heterogeneous signals.

%% file: apps/linear-algebra-with-metric-tensors.tex
\section{Linear Algebra with Metric Tensors}
\label{app:metric}

This appendix collects the algebraic identities involving metric tensors that are used throughout the paper. 
All results are standard in the theory of finite-dimensional inner product spaces; we gather them here for convenience and to fix notation. 
The material in this section is based on~\citeSM{SMstrang2006linear}.

\subsection{Inner Products and Induced Norms}
\label{app:inner_products}

Let $\G \in \pd^d$ be a symmetric positive definite matrix (a \emph{metric tensor}) on $\reall^d$. 
The \emph{$\G$-inner product} and its induced norm are
\begin{equation}
  \langle \x, \y \rangle_\G = \x^\top \G \y,
  \qquad
  \|\x\|_\G = \sqrt{\x^\top \G \x}, \qquad \x, \y \in \reall^d.
  \label{eq:G_inner_product}
\end{equation}
The pair $(\reall^d, \langle\cdot,\cdot\rangle_\G)$ is a finite-dimensional real Hilbert space. 
Since $\G \in \pd^d$, it admits a unique symmetric positive definite square root $\G^{1/2} \in \pd^d$ satisfying $\G^{1/2}\G^{1/2} = \G$, and
\begin{equation}
  \langle \x, \y \rangle_\G = \langle \G^{1/2}\x, \G^{1/2}\y \rangle_{\eye{}} = (\G^{1/2}\x)^\top (\G^{1/2}\y),
  \label{eq:G_inner_product_sqrt}
\end{equation}
so every $\G$-inner product reduces to the standard Euclidean inner product after the change of variables $\tilde{\x} = \G^{1/2}\x$.

For the space of 0-cochains $C^0(\graph;\ns) \cong \reall^L$, $L = \sum_i d_i$, the global metric tensor is $\mathbb{G}_0 = \mathrm{blkdiag}(\{\G_i\}_{i\in\vertexset}) \in \pd^L$, and the $\mathbb{G}_0$-inner product is
\begin{equation}
  \langle \x, \y \rangle_{C^0} = \x^\top \mathbb{G}_0 \y = \sum_{i\in\vertexset} \x_i^\top \G_i \y_i,
  \label{eq:C0_inner_product}
\end{equation}
where the second equality follows from the block-diagonal structure of $\mathbb{G}_0$. 
Similarly, for 1-cochains $C^1(\graph;\ns) \cong \reall^M$, $M = \sum_e d_e$, the global metric is $\mathbb{G}_1 = \mathrm{blkdiag}(\{\G_e\}_{e\in\edgeset})$.

\subsection{G-Adjoint of a Linear Map}
\label{app:adjoint}

Let $\mathbf{A}: (\reall^{d_0}, \langle\cdot,\cdot\rangle_{\G_0}) \to (\reall^{d_1}, \langle\cdot,\cdot\rangle_{\G_1})$ be a linear map represented by a matrix $\mathbf{A} \in \reall^{d_1 \times d_0}$.
Its \emph{$(\G_0, \G_1)$-adjoint} is the unique linear map $\mathbf{A}^\star: (\reall^{d_1}, \langle\cdot,\cdot\rangle_{\G_1}) \to (\reall^{d_0}, \langle\cdot,\cdot\rangle_{\G_0})$ satisfying
\begin{equation}
  \langle \mathbf{A}\x, \y \rangle_{\G_1} = \langle \x, \mathbf{A}^\star\y \rangle_{\G_0}
  \quad \forall\, \x \in \reall^{d_0},\, \y \in \reall^{d_1}.
  \label{eq:adjoint_def}
\end{equation}
Expanding both sides using~\cref{eq:G_inner_product}, $(\mathbf{A}\x)^\top\G_1\y = \x^\top{\mathbf{A}^\star}^\top\G_0\y$ for all $\x,\y$, which gives
\begin{equation}
  \mathbf{A}^\star = \G_0^{-1}\mathbf{A}^\top\G_1.
  \label{eq:adjoint_formula}
\end{equation}
When $\G_0 = \eye{d_0}$ and $\G_1 = \eye{d_1}$, \cref{eq:adjoint_formula} reduces to the standard matrix transpose $\mathbf{A}^\star = \mathbf{A}^\top$. 
The sheaf coboundary adjoint $\B^\star = (\mathbb{G}_0)^{-1}\B^\top\mathbb{G}_1$ is the instance of \cref{eq:adjoint_formula} with $\mathbf{A} = \B$, $\G_0 = \mathbb{G}_0$, and $\G_1 = \mathbb{G}_1$.

A matrix $\mathbf{A} \in \reall^{d \times d}$ is \emph{$\G$-symmetric} (or $\G$-selfadjoint) if $\mathbf{A}^\star = \mathbf{A}$, i.e.,
\begin{equation}
  \G\mathbf{A} = \mathbf{A}^\top\G,
  \label{eq:G_symmetric}
\end{equation}
or equivalently, $\G^{1/2}\mathbf{A}\G^{-1/2}$ is symmetric in the standard sense. 
The sheaf Laplacian $\mathbf{L}_\ns = \B^\star\B$ is $\mathbb{G}_0$-symmetric by construction.

\subsection{G-Orthogonal Projectors}
\label{app:projectors}

A matrix $\mathbf{M} \in \reall^{d \times d}$ is a \emph{$\G$-orthogonal projector} onto a subspace $\mathcal{V} \subseteq \reall^d$ if:
\begin{enumerate}[label=(\roman*)]
  \item \emph{Idempotence}: $\mathbf{M}^2 = \mathbf{M}$;
  \item \emph{$\G$-symmetry}: $\G\mathbf{M} = \mathbf{M}^\top\G$.
\end{enumerate}
These two conditions together imply that $\mathbf{M}$ projects onto $\mathcal{V} = \Im{\M}$ along the $\G$-orthogonal complement $\mathcal{V}^{\perp_\G} = \{\x : \langle\x,\vecv\rangle_\G = 0\, \forall\, \vecv\in\mathcal{V}\}$.

\spara{Construction.} 
Let $\V \in \reall^{d \times r}$ be a matrix whose columns form a $\G$-orthonormal basis for $\mathcal{V}$, i.e., $\V^\top\G\V = \eye{r}$. 
Then
\begin{equation}
  \mathbf{M} = \V\V^\top\G
  \label{eq:G_projector}
\end{equation}
is the $\G$-orthogonal projector onto $\mathrm{span}(\V)$.

Regarding idempotence: $\mathbf{M}^2 = \V\V^\top\G\V\V^\top\G =
  \V(\V^\top\G\V)\V^\top\G = \V\eye{r}\V^\top\G = \mathbf{M}$.

$\G$-symmetry is also verified as follows: 
$\G\mathbf{M} = \G\V\V^\top\G$ and $\mathbf{M}^\top\G = (\V\V^\top\G)^\top\G = \G\V\V^\top\G$.

%% file: apps/proofs.tex
\section{Proofs}
\label{app:proofs}

\begin{proof}[Proof of Prop.~III.1]
    If
    \begin{equation}\label{eq:supp_local_representation_commutativity}
    \restrictionmap{i}{e}\D_i = \D_e\restrictionmapscoeff{i}{e}, \qquad \restrictionmap{j}{e}\D_j = \D_e\restrictionmapscoeff{j}{e} \end{equation}
    then the columns of $\restrictionmap{i}{e}\D_i$ and $\restrictionmap{j}{e}\D_j$ belong to $\mathrm{Im}(\D_e)$.
    Hence,
    \begin{equation}\label{Im(Ae)}
        \mathrm{Im}(\A_e)\subseteq \mathrm{Im}(\D_e).
    \end{equation}
    Since $\D_e\in\stiefel{d_e}{c_e}$, we have $\dim\mathrm{Im}(\D_e)=c_e$, and therefore $\mathrm{rank}(\A_e)\leq c_e$.
    Conversely, if $\mathrm{rank}(\A_e)\leq c_e$, then $\mathrm{Im}(\A_e)$ can be embedded in a $c_e$-dimensional subspace of $\ns(e)$.
    Choosing $\D_e$ as any orthonormal basis of such a subspace, the columns of $\restrictionmap{i}{e}\D_i$ and $\restrictionmap{j}{e}\D_j$ belong to $\mathrm{Im}(\D_e)$, so there exist matrices $\restrictionmapscoeff{i}{e}$ and $\restrictionmapscoeff{j}{e}$ satisfying \eqref{eq:supp_local_representation_commutativity}.
\end{proof}

\begin{proof}[Proof of Prop.~III.2]
Since $\mathrm{Im}(\A_e)\subseteq \mathrm{Im}(\D_e)$ by construction, and
$\D_e$ has orthonormal columns, we have
\begin{equation}
\D_e\D_e^\top\A_e=\A_e\,.
\end{equation}
Using
\begin{equation}
\restrictionmapscoeff{i}{e}=\D_e^\top\restrictionmap{i}{e}\D_i,
\qquad
\restrictionmapscoeff{j}{e}=\D_e^\top\restrictionmap{j}{e}\D_j\,;    
\end{equation}
we have
\begin{equation}
\D_e\restrictionmapscoeff{i}{e}=\D_e\D_e^\top\restrictionmap{i}{e}\D_i=\restrictionmap{i}{e}\D_i\,.
\end{equation}
The same argument gives
\begin{equation}
\D_e\restrictionmapscoeff{j}{e}=\restrictionmap{j}{e}\D_j\,.
\end{equation}
Therefore, both commutativity conditions in \cref{eq:supp_local_representation_commutativity} hold.
\end{proof}

\begin{proof}[Proof of Prop.~III.3]
Let $\vecs\in\gsspace{\graph}{\nscoeff}$.
For every edge $e=(i,j)$,
\begin{equation}
\restrictionmapscoeff{i}{e}\vecs_i=\restrictionmapscoeff{j}{e}\vecs_j\,.
\end{equation}
Using the naturality conditions,
\begin{equation}
\restrictionmap{i}{e}\D_i=\D_e\restrictionmapscoeff{i}{e},
\qquad
\restrictionmap{j}{e}\D_j=\D_e\restrictionmapscoeff{j}{e}\,,
\end{equation}
we obtain
\begin{equation}
\restrictionmap{i}{e}\x_i=\D_e\restrictionmapscoeff{i}{e}\vecs_i=\D_e\restrictionmapscoeff{j}{e}\vecs_j=\restrictionmap{j}{e}\x_j\,.
\end{equation}
Hence, $\x\in\gsspace{\graph}{\ns}$.

Conversely, let $\x\in\gsspace{\graph}{\ns}$ with $\x_i\in\im{\D_i}$ for every node.
Since $\D_i\in\stiefel{d_i}{c_i}$, the representations are uniquely recovered as
\begin{equation}
\vecs_i=\D_i^\top\x_i\,.
\end{equation}
For every edge $e=(i,j)$, naturality together with $\x\in\gsspace{\graph}{\ns}$ gives
\begin{equation}
\D_e\restrictionmapscoeff{i}{e}\vecs_i=\restrictionmap{i}{e}\x_i=\restrictionmap{j}{e}\x_j=\D_e\restrictionmapscoeff{j}{e}\vecs_j\,,
\end{equation}
and since $\D_e$ has orthonormal, hence injective, columns, it can be cancelled on the left, yielding $\restrictionmapscoeff{i}{e}\vecs_i=\restrictionmapscoeff{j}{e}\vecs_j$, i.e., $\vecs\in\gsspace{\graph}{\nscoeff}$.
Therefore, $\mathbb D_0$ is a linear isomorphism between $\gsspace{\graph}{\nscoeff}$ and $\gsspace{\graph}{\ns}\cap\im{\mathbb D_0}$.
\end{proof}

\begin{proof}[Proof of Theorem~IV.1]
Let
\begin{equation}
\mathbb{D}_0 \coloneqq \mathrm{blkdiag}(\{\D_i\}_{i\in\vertexset}),
\quad
\mathbb{D}_1 \coloneqq \mathrm{blkdiag}(\{\D_e\}_{e\in\edgeset})\,.
\end{equation}
The natural transformation condition gives, at the coboundary level,
\begin{equation}
\label{eq:coboundary_intertwining}
\B_{\ns}\mathbb{D}_0 = \mathbb{D}_1\B_\nscoeff\,.
\end{equation}
Indeed, for every edge $e=(i,j)$,
\begin{equation}
\begin{aligned}
(\B_{\ns}\mathbb{D}_0\vecs)_e &= \restrictionmap{i}{e}\D_i\vecs_i-\restrictionmap{j}{e}\D_j\vecs_j\\
&=\D_e\left(\restrictionmapscoeff{i}{e}\vecs_i-\restrictionmapscoeff{j}{e}\vecs_j\right)\\
&=(\mathbb{D}_1\B_\nscoeff\vecs)_e \,.
\end{aligned}
\end{equation}
We now show that the corresponding adjoints also intertwine.
By the metric-compatible lifting assumption,
\begin{equation}
\restrictionmap{i}{e}=\D_e\restrictionmapscoeff{i}{e}\Hmet_i^{-1}\D_i^\top\G_i \,,
\end{equation}
and, since $\G_i,\Hmet_i$ are symmetric, transposing gives
\begin{equation}
\restrictionmap{i}{e}^\top=\G_i\D_i\Hmet_i^{-1}\left(\restrictionmapscoeff{i}{e}\right)^\top\D_e^\top \,.
\end{equation}
Using the adjoint formulas
\begin{equation}
\restrictionmap{i}{e}^{\star}=\G_i^{-1}\restrictionmap{i}{e}^{T}\G_e,
\qquad
\left(\restrictionmapscoeff{i}{e}\right)^{\star}=\Hmet_i^{-1}\left(\restrictionmapscoeff{i}{e}\right)^\top\Hmet_e \,,
\end{equation}
together with the induced-metric identity $\Hmet_e=\D_e^\top\G_e\D_e$, we obtain
\begin{equation}
\begin{aligned}
\restrictionmap{i}{e}^{\star}\D_e &=\G_i^{-1}\restrictionmap{i}{e}^{T}\G_e\D_e \\
&=\D_i\Hmet_i^{-1}\left(\restrictionmapscoeff{i}{e}\right)^\top\underbrace{\D_e^\top\G_e\D_e}_{=\,\Hmet_e}
=
\D_i\left(\restrictionmapscoeff{i}{e}\right)^{\star}\,.
\end{aligned}
\end{equation}
Therefore, at the global level,
\begin{equation}
\label{eq:adjoint_intertwining}
\B_{\ns}^{\star}\mathbb{D}_1=\mathbb{D}_0\B_\nscoeff^{\star}\,.
\end{equation}

Using \eqref{eq:coboundary_intertwining} and \eqref{eq:adjoint_intertwining}, we obtain the result $\forall\,\vecs\in C^0(\graph;\nscoeff)$
\begin{equation}
\Lsh\mathbb{D}_0=\B_{\ns}^{\star}\B_{\ns}\mathbb{D}_0=\B_{\ns}^{\star}\mathbb{D}_1\B_\nscoeff=\mathbb{D}_0\B_\nscoeff^{\star}\B_\nscoeff=\mathbb{D}_0\Lshrep\,.
\end{equation}
\end{proof}

\begin{proof}[Proof of Cor.~IV.2]
If $\Lshrep\vecv=\mu\vecv$, then
\begin{equation}
\Lsh\,\mathbb D_0\vecv=\mathbb D_0(\Lshrep\vecv)=\mu\,\mathbb D_0\vecv\,,
\end{equation}
so $\mathbb D_0\vecv$ is an eigenvector of $\Lsh$ associated with the same eigenvalue, provided $\mathbb D_0\vecv\neq \zeros$.

Finally, using the definition $\Hmet_i=\D_i^\top\G_i\D_i$, we obtain
\begin{equation}
\begin{aligned}
\Eprod{\mathbb D_0\vecv}{\mathbb D_0\w}{\mathcal{C}^{0}(\graph;\ns)}&=\sum_{i\in\vertexset}(\D_i\vecv_i)^\top\G_i(\D_i\w_i)\\
&=
\Eprod{\vecv}{\w}{\mathcal{C}^{0}(\graph;\nscoeff)}\,.
\end{aligned}
\end{equation}
Thus $\mathbb D_0$ is an isometric embedding between the corresponding $0$-cochain Hilbert spaces.
\end{proof}

\begin{proof}[Proof of Cor.~V.1]
Since
\begin{equation}
\mathbf{L}_{\ns}\mathbb D_0 = \mathbb D_0\mathbf{L}_\nscoeff\,,
\end{equation}
it follows by induction that
\begin{equation}
\mathbf{L}_{\ns}^{\,q}\mathbb D_0 = \mathbb D_0\mathbf{L}_\nscoeff^{\,q},\qquad q\ge0.
\end{equation}
Multiplying by the coefficients $a_q$ and summing over $q=0,\ldots,Q$ yields
\begin{equation}
    h(\mathbf{L}_{\ns})\,\mathbb D_0=\mathbb D_0\,h(\mathbf{L}_\nscoeff)\,.    
\end{equation}
\end{proof}

\begin{proof}[Proof of Thm.~VI.1]

Since $\x\in\mathcal B_{\mathcal K}$, there exists a unique vector $\boldsymbol{\alpha}\in\mathbb R^{b_{\mathcal K}}$ such that
\begin{equation}
\x=\mathbf V_{\mathcal K}\boldsymbol{\alpha}.
\end{equation}
Hence,
\begin{equation}
\x_{\mathcal S}=\boldsymbol{\Psi}_{\mathcal S}\mathbf V_{\mathcal K}\boldsymbol{\alpha}.
\end{equation}

The coefficient vector $\boldsymbol{\alpha}$ is uniquely determined if and only if the matrix $\boldsymbol{\Psi}_{\mathcal S}\mathbf V_{\mathcal K}$ has full column rank, proving 
\begin{equation}
\label{eq:sampling_rankSM}
\operatorname{rank}\!\left(\boldsymbol{\Psi}_{\mathcal S} \mathbf V_{\mathcal K} \right)=b_{\mathcal K},
\end{equation}

Since $\mathcal B_{\mathcal K} = \operatorname{Im}(\mathbf V_{\mathcal K})$, condition \eqref{eq:sampling_rankSM} is equivalent to requiring that no nonzero bandlimited signal belongs to $\ker(\boldsymbol{\Psi}_{\mathcal S})$, yielding $\mathcal B_{\mathcal K}\cap\ker(\boldsymbol{\Psi}_{\mathcal S})=\{\mathbf0\}.$.

Finally, when \eqref{eq:sampling_rankSM} holds,
\begin{equation}
\widehat{\boldsymbol{\alpha}}=\left(\boldsymbol{\Psi}_{\mathcal S}\mathbf V_{\mathcal K}\right)^{\dagger}\x_{\mathcal S},
\end{equation}
and substituting this expression into $\x=\mathbf V_{\mathcal K}\boldsymbol{\alpha}$ gives
\begin{equation}
\label{eq:sampling_reconstructionSM}
\widehat{\x}=\mathbf V_{\mathcal K}\left(\boldsymbol{\Psi}_{\mathcal S}\mathbf V_{\mathcal K}\right)^{\dagger}\x_{\mathcal S}.
\end{equation}
\end{proof}

\begin{proof}[Proof of Cor.~VI.2]
Perfect recovery requires
\begin{equation}
\operatorname{rank}\!\left(\boldsymbol{\Psi}_{\mathcal S}\mathbf V_{\mathcal K}\right)=b_{\mathcal K}\,.
\end{equation}
Since $\boldsymbol{\Psi}_{\mathcal S}\mathbf V_{\mathcal K}$ has $|\mathcal S|$ rows, its rank cannot exceed $|\mathcal S|$.
Therefore, $|\mathcal S|\geq b_{\mathcal K}$.
\end{proof}

\begin{proof}[Proof of Cor.~VI.3]
Since $\mathbf L_\ns\mathbb D_0(\vecv)=\mathbb D_0(\mathbf L_\nscoeff\vecv)\in\im{\mathbb D_0}$ for every $\vecv$, the subspace $\im{\mathbb D_0}$ is $\mathbf L_\ns$-invariant.
Since $\mathbf L_\ns$ is $\mathbb G_0$-self-adjoint, its $\mathbb G_0$-orthogonal complement is invariant as well.
Hence $\boldsymbol\Pi_\delta$ commutes with $\mathbf L_\ns$, with every spectral projector $\mathbf P_k$, and therefore with $\mathbf B_\mathcal K$, giving the stated splitting.
For the explicit form, $\mathbb D_0$ is injective and maps $\mathcal E_k$ isometrically into $\mathcal F_k$, so $m'_k\le m_k$.
The $\mathbf U_k$ bases are mutually $\mathbb H_0$-orthogonal, corresponding to eigenspaces of $\mathbf L_\nscoeff$ for distinct eigenvalues, whence $\mathrm{Im}(\mathbf B_\mathcal K^\parallel)=\mathrm{Im}(\mathbb D_0\mathbf U_\mathcal K)$.
\end{proof}

%% file: apps/spectral-multiplicity.tex
\section{Multiplicity of the spectrum}\label{app:spectral-multiplicity}

Unlike graph Laplacians, whose eigenvalues are generically simple, the spectrum of a sheaf Laplacian often exhibits eigenvalues with multiplicity greater than one. 
This phenomenon is not accidental, but reflects the geometric structure encoded by the restriction maps. 
Whenever the latter transport information coherently across the sheaf, the associated eigenspaces acquire additional degrees of freedom, leading to repeated eigenvalues. 
In this subsection we illustrate two representative mechanisms generating spectral multiplicity.

\begin{example}[Hierarchical Stiefel embeddings]
\label{ex:stiefel_multiplicity}

Suppose that the restriction maps $\restrictionmapscoeff{i}{e}$ are Stiefel matrices belonging to $\stiefel{c_e}{c_i}$, acting as isometric embeddings ($c_i\leq c_e$).
Each edge can then be naturally oriented from the lower-dimensional stalk towards the higher-dimensional one.
Assume furthermore that every node is reachable through an oriented path starting from a node $N$ whose stalk has minimum dimension $c_N=h$.
Then
\begin{equation}
    \dim\!\bigl(\ker(\mathbf{L}_{\nsrep})\bigr)=h,
\end{equation}
and every global section is uniquely determined by a single $h$-dimensional vector assigned to the stalk at node $N$, which is propagated consistently throughout the sheaf via the Stiefel restriction maps~\citeSM{SMdacunto2026networkscausalabstractionssheaftheoretic}.

Consequently, whenever $h>1$, the zero eigenvalue has multiplicity $h$. 
Any orthonormal basis returned by an eigensolver is simply one among infinitely many orthonormal bases spanning the same null space, and no individual basis vector carries an intrinsic meaning.
\end{example}

\begin{example}[Connection graph]
\label{ex:connection_graph_multiplicity}

The previous example concerns only the null eigenspace.
A much stronger form of multiplicity arises in the homogeneous case $c_i=c_e=c$, where the restriction maps are special orthogonal matrices $\restrictionmapscoeff{i}{e}=\myO_{ij}\in\specialO{c}$ defining a connection graph.
If the connection is consistent, namely the transport maps compose to the identity around every cycle of $\graph$, there exist orthogonal matrices $\{\mathbf{O}_i\}_{i\in\vertexset}$ such that
\begin{equation}
    \myO_{ij}=\mathbf{O}_i^\top\mathbf{O}_j\,,    
\end{equation}
and the sheaf Laplacian admits the factorization
\begin{equation}
\mathbf{L}_{\nsrep} = \bdO^\top \bigl( \mathbf{L}_{\graph}\otimes\eye{c}\bigr)\bdO\,,
\end{equation}
where
$\bdO=\mathrm{blkdiag}(\mathbf{O}_1,\ldots,\mathbf{O}_N)$ and $\mathbf{L}_{\graph}$ denotes the scalar graph Laplacian \citeSM{SMchung2014ranking}.

In this case, every eigenvalue of the graph Laplacian is repeated exactly $c$ times.
The graph topology and the stalk geometry completely decouple, and each graph frequency generates a $c$-dimensional eigenspace of the sheaf Laplacian.
\end{example}

\cref{ex:stiefel_multiplicity,ex:connection_graph_multiplicity} illustrate two different mechanisms leading to spectral multiplicity.
In \cref{ex:stiefel_multiplicity}, multiplicity originates from the hierarchical embedding structure of the restriction maps and affects only the null space of the sheaf Laplacian.
In \cref{ex:connection_graph_multiplicity}, it results from a geometric symmetry that propagates throughout the entire spectrum.
More generally, these examples suggest that spectral multiplicity is an intrinsic consequence of the geometric structure induced by the restriction maps.